\documentclass[journal,12pt,onecolumn,draftclsnofoot,]{IEEEtran}

\ifCLASSINFOpdf

\else

\fi

\usepackage{array}

\usepackage[utf8]{inputenc}
\usepackage[T1]{fontenc}

\usepackage{graphicx} 
\usepackage{subcaption} 
\usepackage{color}
\usepackage[utf8]{inputenc}
\usepackage[mathscr]{euscript}
\usepackage{filecontents}
\usepackage{amsfonts}
\usepackage{filecontents}
\usepackage{bm}
\usepackage{algorithm}
\usepackage[noend]{algpseudocode}
\usepackage{enumitem}   
\usepackage{cite}

\usepackage{amsthm}
\usepackage{mdwmath}
\usepackage{amsmath}
\usepackage{mdwtab}
\usepackage[utf8]{inputenc}
\usepackage{multicol}
\usepackage{eqparbox}
\usepackage{wrapfig,lipsum,booktabs}
\usepackage{graphicx}
\usepackage{epstopdf} %converting to PDF

\usepackage{color}
\usepackage{placeins}
\usepackage{float}
\usepackage[square,sort,comma,numbers]{natbib}

\usepackage{lipsum}

\newcommand\blfootnote[1]{%
  \begingroup
  \renewcommand\thefootnote{}\footnote{#1}%
  \addtocounter{footnote}{-1}%
  \endgroup
}

\begin{document}
%
% paper title
% can use linebreaks \\ within to get better formatting as desired
\title{Resource Allocation in Heterogenous Full-duplex OFDMA Networks: Design and Analysis
}

% author names and affiliations
% use a multiple column layout for up to three different
% affiliations
\author{\IEEEauthorblockN{Peyman Tehrani\IEEEauthorrefmark{1},
Farshad Lahouti\IEEEauthorrefmark{2},
Michele Zorzi\IEEEauthorrefmark{3}}

\IEEEauthorblockA{\IEEEauthorrefmark{1} Computer Science Department,
University of California, Irvine, USA} \\
\IEEEauthorblockA{\IEEEauthorrefmark{2}Electrical Engineering Department, California Institute of Technology, USA} \\
\IEEEauthorblockA{\IEEEauthorrefmark{3}Department of Information Engineering, University of Padova, Italy \\
Emails: peymant@uci.edu, lahouti@caltech.edu, zorzi@dei.unipd.it}
}

%\author{\IEEEauthorblockN{Peyman Tehrani}
%\IEEEauthorblockA{University of Tehran\\
%peymantehrani@ut.ac.ir\\}
%}

\maketitle

\vspace{-6mm}
\begin{abstract}
%\boldmath
Recent studies indicate the feasibility of full-duplex (FD) bidirectional wireless communications. Due to its potential to increase the capacity, analyzing the performance of a cellular network that contains full-duplex devices is crucial. In this paper, we consider maximizing the weighted sum-rate of downlink and uplink of an FD heterogeneous OFDMA network where each cell consists of an imperfect FD base-station (BS) and a mixture of half-duplex and imperfect full-duplex mobile users. To this end, first, the joint problem of sub-channel assignment and power allocation for a single cell network is investigated. Then, the proposed algorithms are extended for solving the optimization problem for an FD heterogeneous network in which intra-cell and inter-cell  interferences are taken into account. Simulation results demonstrate that in a single cell network, when all the users and the BSs are perfect FD nodes, the network throughput could be doubled. Otherwise, the performance improvement is limited by the inter-cell interference, inter-node interference, and self-interference. We also investigate  the effect of the percentage of FD users on the network performance in both indoor and outdoor scenarios, and analyze the effect of the self-interference cancellation capability of the FD nodes on the network performance.

\textbf{\textit{Index Terms}}: Full-duplex, self-interference, resource allocation, OFDMA, femto cell, heterogeneous.
\end{abstract}

\IEEEpeerreviewmaketitle
\pagebreak

\vspace{-8mm}
\section{Introduction}

\vspace{-8mm}
\blfootnote{This paper has been presented in part at the IEEE International Conference on Communication (ICC), Kuala Lumpur, Malaysia, May 2016.}

In wireless communications, separation of transmission and reception in time or frequency has been the standard practice so far. However, through simultaneous transmission and reception in the same frequency band, wireless full-duplex has the potential to double the spectral efficiency. Due to this substantial gain, full-duplex technology has recently attracted noticeable interest in both academic and industrial worlds. The main challenge in full-duplex (FD) bidirectional communication is self-interference (SI) cancellation. In recent years, many attempts have been made to cancel the self-interference signal \cite{korpi2014widely, kaufman2013analog, ahmed2014all, duarte2012experiment}. In \citep{bharadia2013full}, it is shown that $110$ dB SI cancellation is achievable, and by jointly exploiting analog and digital techniques, SI may be reduced to the noise floor.

A full-duplex physical layer in cellular communications calls for a re-design of higher layers of the protocol stack, including scheduling and resource allocation algorithms. In \cite{goyal2013analyzing}, the performance of an FD-based cellular system is investigated and an analytic model to derive the average uplink and downlink channel rate is provided. A resource allocation problem for an FD heterogeneous orthogonal frequency-division multiple access (OFDMA) network is considered in \cite{sultanmode}, in which the macro base station (BS) and small cell access points operate in either  FD or half-duplex (HD) MIMO mode, and all mobile nodes operate in HD single antenna mode. In \cite{di2014radio}, using matching theory, a sub-channel allocation algorithm for an FD OFDMA network is proposed. In both \cite{sultanmode} and \cite{di2014radio} only a single sub-channel is assigned to each of the uplink users in which they transmit with constant power. Resource allocation solutions are proposed in \cite{nam2015joint} and \cite{namradio} for FD OFDMA networks with perfect FD nodes (SI is canceled perfectly).

Recent research reports investigate resource allocation in  multi-cell FD networks. In \cite{malik2017suboptimal}, a sub-optimal resource management algorithm is presented for the sum rate maximization of a small multi-cell system, including FD base stations and HD mobile users. In \cite{randrianantenaina2017interference}, the problem of maximizing a network-wide
rate-based utility function subject to uplink (UL) and downlink (DL) power constraints is studied in a flexible duplex system, in which UL/DL channels are allowed to have partial overlap via fine-tuned bandwidth allocation. For simplicity, it is assumed that the number of sub-channels and the users are exactly the same. In \cite{sekander2016decoupled}, the problem of decoupled UL-DL
user association, which allows users to associate with different BSs for UL and DL transmissions, is investigated in a multi-tier FD  Network. In \cite{aquilina2017weighted}, weighted sum rate maximization in a FD multi-user multi-cell MIMO network  is studied.  A user scheduling and power allocation method for ultra-dense FD small-cell networks is presented in \cite{atzeni2016performance}. In \cite{sekander2016decoupled} \cite{aquilina2017weighted} and \cite{atzeni2016performance}, the sub-channel allocation problem is not investigated since a single channel network is assumed. The most related work to the current research is \cite{yun2016intra}, in which, a radio resource management solution for an OFDMA FD heterogeneous cellular network is presented. The algorithm jointly assigns the transmission mode, and the user(s) and their transmit power levels for each frequency resource block to optimize the sum of the downlink and uplink rates. The users are assumed to use a single class of service. A sub-optimal resource allocation algorithm is then proposed which takes into account both intra-cell and inter-cell interferences. The sub-optimal power adjustment algorithm is designed under the assumption of high SINR, where the rate of an FD-FD or FD-HD link is independent of power variations. 

In this paper, we consider a general resource allocation problem in a heterogeneous OFDMA-based network consisting of imperfect FD macro BS and femto BSs and both HD and imperfect FD users. We aim to maximize the downlink and uplink weighted sum-rate of femto users  while protecting the macro users rates. The weights allow for users to utilize differentiated classes of service, accommodate both frequency or time division duplex for HD users, and prioritize uplink or downlink transmissions. To be more realistic, imperfect SI cancellation in FD devices is assumed and FD nodes suffer from their SI. A contribution of the current work is to consider the presence of a mixture of FD and HD users, which enables us to quantify the percentage of FD users needed to capture the full potential of FD technology in wireless OFDMA networks. We also analyze the effect of the SI cancellation level on the network performance, which to our knowledge has not been studied in prior works. We will show that when the SI cancellation capability is worse than a specified threshold, then the throughput of an all FD user network would not be larger than the throughput of an all HD user network. Moreover, we will analyze this threshold theoretically and compare its outcome with simulation results. 

The remainder of this paper is organized as follows. In Section \ref{system_model}, the basic system model of a single cell FD network is given and the optimization problem is formulated. In Section \ref{subchannel allocation}, a sub-channel allocation algorithm for selecting the best pair in each sub-channel is presented. Power allocation is considered in Section \ref{power_allocation}. A theoretical approach for deriving the SI cancellation coefficient threshold is proposed in Section \ref{Coefficient_Threshold}. In Section \ref{Two-tier}, the optimization problem for an FD heterogeneous network is presented. Numerical results for the proposed methods are shown in Section \ref{simulation_results}. Finally, the paper is concluded in Section \ref{conclusion}. 

\vspace{-3mm}
\section{System Model And Problem Statement} \label{system_model}
We consider a single cell network that consists of a full-duplex base-station (BS) and a total of $K$ half-duplex and full-duplex users. For communications between the nodes and the BS, we assume that an OFDMA system with $N$ sub-channels is used. All sub-carriers are assumed to be perfectly synchronized, and so there is no interference between different sub-channels. Since the base-station operates in full-duplex mode, it can transmit and receive simultaneously in each sub-channel. In each timeslot the base-station is to properly allocate the sub-channels to the downlink or uplink of appropriate users and also determine the associated transmission power in an optimized manner. We assume that the base-station and the FD users are imperfect full-duplex nodes that suffer from self-interference. We define a self-interference cancellation coefficient to take this into account in our model and denote it by $0\leq\beta\leq1$, where $\beta=0$  indicates that SI is canceled perfectly and $\beta=1$  means no SI cancellation. For simplicity, we assume the same self-interference cancellation coefficient for BS and FD users, but consideration of different coefficients would be possible. In this paper, the goal is to maximize the weighted sum-rate of downlink and uplink users with a total power constraint at the base-station and a transmission power constraint for each user.

\begin{figure}[t]
\centering
  \includegraphics[width=9cm, height=6cm]{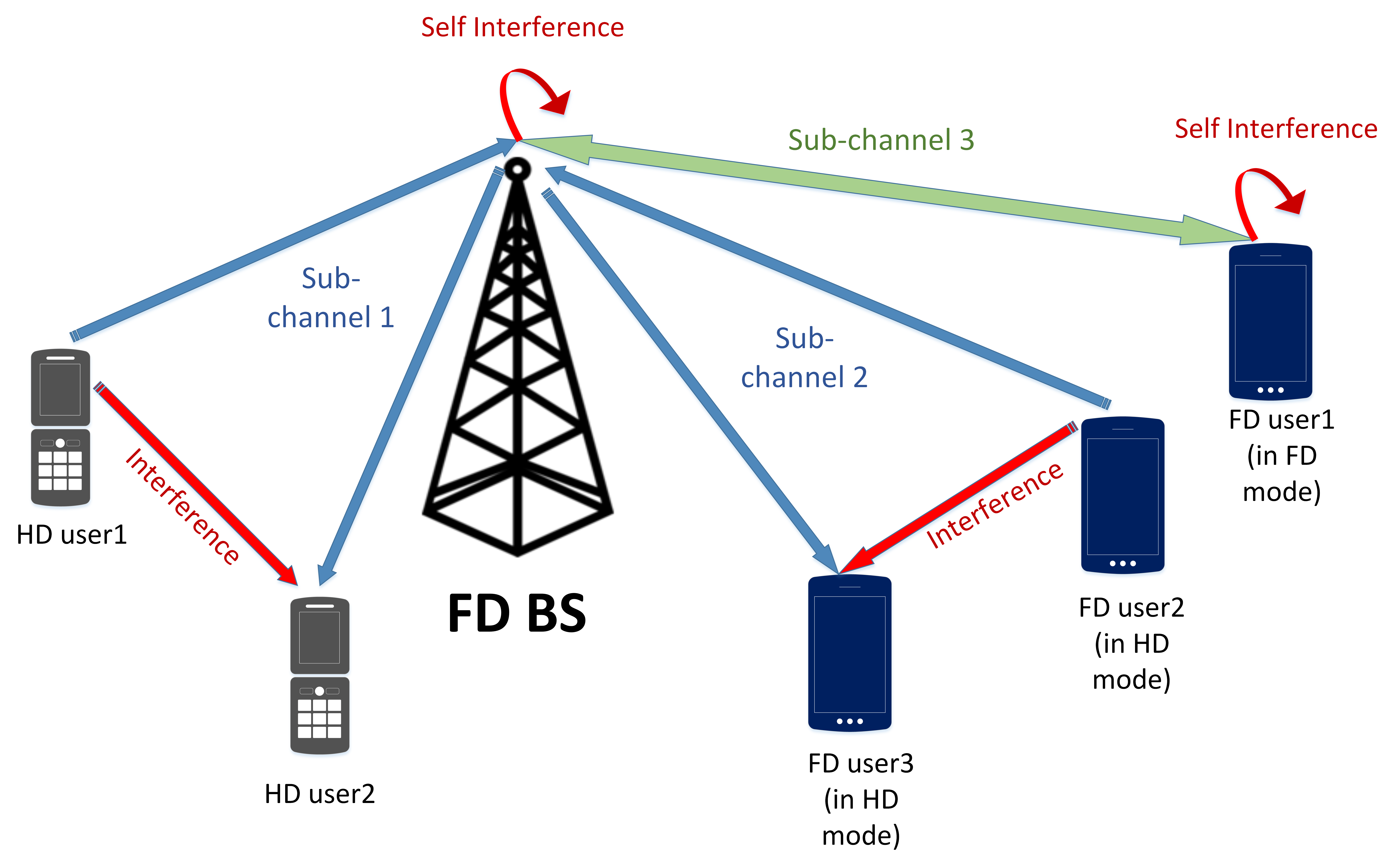}
  \vspace{-6mm}
  \caption{\small{A single cell OFDMA full-duplex network that contains an imperfect full-duplex base-station and multiple half-duplex and full-duplex mobile nodes. Due to the full-duplex nature of this network, the base-station suffers from its self interference, and the uplink nodes cause interference to their co-channel downlink nodes.}}
\vspace{-9mm}
\end{figure}

 We define the downlink weighted sum-rate as
\vspace{-1mm}
\begin{equation}\label{DLRateQ}
R_d=\sum_{k=1}^{K}\sum_{n\in S_{k,d}}w_k\log \bigg(1+ \frac{g_k(n)p_{k,d}(n)}{N_k+I_{k,j}(n)p_{j,u}(n)} \bigg)
\vspace{-0mm}
\end{equation}

And the uplink weighted sum-rate as 
\vspace{-1mm}
\begin{equation}
R_u=\sum_{j=1}^{K}\sum_{n\in S_{j,u}}v_j\log \bigg(1+ \frac{g_j(n)p_{j,u}(n)}{N_0+\beta p_{k,d}(n)}\bigg)
\end{equation}  
The variables used in the above equations are introduced in Table I. We assume here that the channel is reciprocal, i.e., uplink and downlink channel gains are the same. We further assume that the receiver noise powers in different sub-channels are the same. The term $I_{k,j}(n)p_{j,u}(n)$ in \eqref{DLRateQ} denotes the interference: When user $k$ is a FD device and both downlink and uplink of sub-channel $n$ are allocated to it $(j=k)$, $I_{k,j}(n)=\beta$, else $I_{k,j}(n)=g_{k,j}(n)$ is the channel gain between uplink user $j$ and downlink user $k$. We assume that the base-station knows all the channel gains, the noise powers, and the SI cancellation coefficient and weights assigned to the downlink and uplink of all users. 

Let $P_0$ and $P_k$ denote the maximum available transmit power for the base-station and for user $k$, respectively. Then the proposed design optimization problem, denoted by \textbf{\textit{P}1}, can be formulated as follows

\begin{table}[t!]
 \caption{Main Parameters and Variables}
 \centering
 \begin{tabular}{|c | m{24em}|}

 \hline
 $w_k$ & weight assigned to the downlink  of user $k$ \\
 \hline
 $v_k$ & weight assigned to the uplink of user $k$ \\
 \hline
  $p_{k,d}(n)$ &  transmission power from BS to user $k$ on sub-channel $n$ \\
 \hline
  $p_{j,u}(n)$ & transmission power from user $j$ to BS on sub-channel $n$ \\
 \hline
  $ N_k$  & Gaussian noise variance at the receiver of user $k$ \\
 \hline
  $N_0$ & Gaussian noise variance at the base-station receiver\\
 \hline
  $S_{k,d}$ & set of sub-channels allocated to user $k$ for downlink \\
 \hline
  $S_{j,u}$ & set of sub-channels allocated to user $j$ for uplink  \\
 \hline
  $\beta$ & self-interference cancellation coefficient \\
 \hline
 $g_{k}(n)$ & channel gain between BS and user $k$ on sub-channel $n$ \\
 \hline
 $g_{k,j}(n)$ & channel gain between users $j$ and $k$ on sub-channel $n$ \\
 \hline
 $I_{k,j}(n)$ & equal to $\beta$ when $j=k$, and to $g_{k,j}(n)$ otherwise\\
 \hline
 $P_0$ & maximum available transmit power at BS\\
 \hline
 $P_k$ & maximum available transmit power at user $k$ \\
 \hline
 \end{tabular}
\end{table}
\vspace{-9mm}
%\begin{align}
%\textbf{\textit{P}1}: \operatorname*{maximize}_{p_{k,d},p_{j,u},S_{j,u},S_{k,d},  \forall k,j} \qquad R_d+R_u \\
%\text{subject to }   \sum_{k=1}^{K}\sum_{n\in S_{k,d}}p_{k,d}(n) \leq P_0\\
%\sum_{n\in S_{j,u}}p_{j,u}(n) \leq P_j \quad \forall j\\
%p_{j,u}(n),p_{k,d}(n)\geq 0 \quad \forall j,k,n\\
% S_{i,d}\cap S_{j,d}=\phi, S_{i,u}\cap S_{j,u}=\phi \quad \forall i\neq j\\
%\cup_{j=1}^{K}  \ S_{j,u} \subseteq \{1,2,...,N\},\cup_{k=1}^{K} \ S_{k,d} \subseteq \{1,2,...,N\} \\
%S_{k,u} \cap S_{k,d}=\phi \quad \text{if user} \ k \  \text{is} \ \text{HD}
%\end{align}

\begin{equation}
\textbf{\textit{P}1}: \operatorname*{maximize}_{p_{k,d},p_{j,u},S_{j,u},S_{k,d},  \forall k,j} \qquad R_d+R_u 
\end{equation}
\begin{equation}
\text{subject to }   \sum_{k=1}^{K}\sum_{n\in S_{k,d}}p_{k,d}(n) \leq P_0
\end{equation}
\begin{equation}
\sum_{n\in S_{j,u}}p_{j,u}(n) \leq P_j \quad \forall j
\end{equation}
\begin{equation}
p_{j,u}(n),p_{k,d}(n)\geq 0 \quad \forall j,k,n
\end{equation}
\begin{equation}
 S_{i,d}\cap S_{j,d}=\phi, S_{i,u}\cap S_{j,u}=\phi \quad \forall i\neq j
 \end{equation}
 \begin{equation}
\cup_{j=1}^{K}  \ S_{j,u} \subseteq \{1,2,...,N\},\cup_{k=1}^{K} \ S_{k,d} \subseteq \{1,2,...,N\} 
\end{equation}
\begin{equation}
S_{k,u} \cap S_{k,d}=\phi \quad \text{if user} \ k \  \text{is} \ \text{HD}
\end{equation}

where (4) and (5) indicate the power constraint on the BS and the users, respectively. Constraint (6) shows the non-negativity feature of powers; (7) come from the fact that a sub-channel cannot be allocated to two distinct users simultaneously; (8) indicate that we have no more than $N$ sub-channels, and the last constraint accounts for the half-duplex nature of the HD users.

The general resource allocation problem presented is combinatorial in nature because of the channel allocation issue and addressing it together with power allocation in an optimal manner is challenging, especially as the number of users and sub-channels grow. Moreover, the non-convexity of the rate function makes the power allocation problem itself challenging even for a fixed sub-channel assignment. Here, we invoke a two step approximate solution. First, we  determine the allocation of downlink and uplink sub-channels to users and then determine the transmit power of the users and the base-station on their allocated sub-channels. In other words, we first specify the sets $S_{k,d}$ and $S_{j,u}$ and then determine the variables $p_{j,u}(n)$, $p_{k,d}(n)$. In the next Section, we introduce our sub-channel allocation algorithm. 

\section{Sub-channel Allocation} \label{subchannel allocation}
%In general, sub-channel assignment is a combinatorial problem and finding its optimal solution requires %$O(K^{2N})$ operations, which is exponential in $N$. Here, we propose a sub-optimal algorithm that its %complexity is polynomial in $N$ and its solution is close enough to the optimal solution.
 
The sub-channel allocation problem, denoted by \textbf{\textit{P}2}, can be formulated as follows 
\begin{align*}
\textbf{\textit{P}2}: & \operatorname*{maximize}_{S_{j,u},S_{k,d}} \qquad &R_d+R_u\\
& \text{subject to} \qquad &\text{(7)-(11)}
\end{align*}
To solve the problem \textbf{\textit{P}2}, we should first solve the following power allocation problem, denoted by \textbf{\textit{P}3}, to maximize the weighted sum-rate in a single sub-channel and for a fixed pair of uplink and downlink users. Since a single sub-channel is being considered in \textbf{\textit{P}3}, we have dropped the variable $n$ in the notation.
\begin{align*}
\textbf{\textit{P}3}: \operatorname*{max}_{p_{k,d},p_{j,u}} L(p_{k,d},p_{j,u})=&w_k\log(1+\frac{g_kp_{k,d}}{N_k+I_{k,j}p_{j,u}})+v_j\log(1+ \frac{g_jp_{j,u}}{N_0+\beta p_{k,d}})
\end{align*}
\begin{align}
&0\leq p_{k,d}\leq P_{max1} \\
&0\leq p_{j,u}\leq P_{max2}
\end{align}
Here, $P_{max1}$ and $P_{max2}$ are the maximum allowable transmit powers.
\newtheorem{prop}{Proposition}
\begin{prop}
For a fixed downlink user $k$ and uplink user $j$, the optimal pair of powers $(p_{k,d}^*,p_{j,u}^*)$ that optimizes \textbf{\textit{P}3} belongs to the following set.
{\small
\begin{align*}
\textbf{S}=\lbrace(0,P_{max2}),(P_{max1},0),(P_{max1},P_{max2}),&(p_{k,d}^a,P_{max2}),(P_{max1},p_{j,u}^a)\rbrace
\end{align*}
\text{ where} 
\vspace{-3mm}
\begin{align}
p_{k,d}^a=\frac{-B-\sqrt{B^2-4AC}}{2A} , p_{j,u}^a=\frac{-E-\sqrt{E^2-4DF}}{2D}
\end{align}
\text{and} 
\vspace{-3mm} 
\begin{align}
A&=w_k g_k \beta^2, B=2w_kN_0g_k\beta+(w_k-v_j)\beta g_kg_jp_{j,u} \\
C&=w_kg_kN_0^2+w_kg_kg_jN_0p_{j,u}-v_jN_kg_jp_{j,u}\beta -v_jg_j\beta I_{k,j}p_{j,u}^2 \\
D&=v_j g_j I_{k,j}^2,E=2v_jN_kg_jI_{k,j}+(v_j-w_k)I_{k,j} g_kg_jp_{k,d} \\
F&=v_jg_jN_k^2+v_jg_kg_jN_kp_{k,d}-w_kN_0g_kp_{k,d}I_{k,j}-w_kg_k\beta I_{k,j}p_{k,d}^2 
\end{align}
}
\end{prop}
\begin{proof}
Computing the derivative with respect to $p_{k,d}$ and setting it to zero we have:
\begin{equation*}
 \frac{\partial L}{\partial p_{k,d}}=0\Longrightarrow Ap_{k,d}^2+Bp_{k,d}+C=0
\end{equation*}
where $A$, $B$ and $C$ are defined above.
It is evident that $A\geq 0$, and if $w_k\geq v_j$ then $B\geq 0$. When $A,B\geq 0$ the  above quadratic equation  either  has  no zeros in $[0,P_{max1}]$ or has only one zero where the function changes sign from $-$ to $+$ indicating a local minimum for $L$. Therefore, in both cases the maximum is attained at a boundary point $0$ or $P_{max1}$. But when $w_k \leq v_j$, $B$ could be negative, and the smaller root of the quadratic equation $p_{k,d}^a$ could be positive. In this case, the maximum is attained  at $P_{max1}$ or $p_{k,d}^a$. By similar analysis for $p_{j,u}$ one sees that if $v_j\geq w_k$ then the maximum is attained at a boundary point $0$ or $P_{max2}$ and when $w_k\geq v_j$ the maximum is attained at $P_{max2}$ or $p_{j,u}^a$.
As a result, when $B\geq0$ the optimal transmission powers belong to the 
following set,
\small{
\begin{equation*}
P_{opt1}=\left\{(0,P_{max2}),(P_{max1},0),(P_{max1},P_{max2}),(P_{max1},p_{j,u}^a)\right\}.
\end{equation*}
}
\normalsize
Otherwise, if $B<0$, they belong to the set below
{\small 
\begin{equation*}
P_{opt2}=\left\{(0,P_{max2}),(P_{max1},0),(P_{max1},P_{max2}),(p_{k,d}^a,P_{max2})\right\}.
\end{equation*}
}
The cases {\small$(0,p_{j,u}^a)$} and {\small$(p_{k,d}^a,0)$} cannot be the optimal solutions of  \textbf{\textit{P}3} , because they are dominated by {\small$(0,P_{max2})$} and {\small$(P_{max1},0)$} which give a larger $L$.
\normalsize
Therefore, optimal powers could be found by checking the members of the set \textbf{S} and picking the one that corresponds to the largest $L$.
\end{proof}
Based on the above Proposition one can find the best uplink-downlink pair in each sub-channel by choosing the one with the largest value of $L$. This involves only $O(K^2)$ operations.
Now we can  present our sub-channel allocation algorithm to solve Problem \textbf{\textit{P}2}, in which we employ a sub-optimum power allocation scheme.
First, for each sub-channel $n$, we find the best channel gain among all users and denote it by $\tilde{g}({n})=\operatorname*{arg\,max}_k g_k(n) $. Then, we sort the sub-channels based on the value of $\tilde{g}({n})$. In other words. we find a sub-channel permutation $\{a_1,...,a_N\}$ such that
$\tilde{g}({a_1}) \geq \tilde{g}({a_2}) \geq . . .\geq \tilde{g}({a_N})$. Then, starting from sub-channel $a_1$, we seek $k$ and $j$ that maximize $L$. At the first iteration, we set $P_{max1}=P_0$ , $P_{max2}=P_k$ and for iteration $l\geq 2$ set $P_{max1}= \frac{P_0}{d_0(l)}$ and $P_{max2}=\frac{P_k}{d_k(l)}$ where $d_0(l)$ and $d_k(l)$ indicate the number of sub-channels to be allocated to the BS's downlink transmission and to user $k$'s uplink transmission, respectively, in the $l$th iteration. The proposed sub-channel allocation algorithm is summarized below.

\begin{table}[h!]
 \centering
 \begin{tabular}{m{30em}}
 \hline
 \textbf{Algorithm 1: Sub-channel Allocation Algorithm} \\
 \hline\hline
1.\textbf{for} $n=1$ to $N$ \textbf{do} \\ 
2.\quad $\tilde{g}({n})=\operatorname*{max}_k g_k(n) $
  \\
3.\textbf{end for} \\
4.Find a sub-channel permutation $\{ a_1,...,a_N \}$, $a_i\in\{1,...,N\}$, $a_i\neq a_j$ such that\\
\hspace{0.5cm} $\tilde{g}({a_1}) \geq \tilde{g}({a_2}) \geq . . .\geq \tilde{g}({a_N})$ \\
5. set $d_k(l)=1$  for $0\leq k \leq K$  and $1 \leq l \leq N$ \\
6.\textbf{for} $l=1$ to $N$ \textbf{do} \\
7.\quad Set $P_{max1}=\dfrac{P_0}{d_0(l)}$  and  $P_{max2}=\dfrac{P_k}{d_k(l)}      \forall k$  \\
8. \quad \textbf{for} $k=1$ to $K$ \textbf{do} \\
9. \quad \quad 	\textbf{for} $j=1$ to $K$ (if $k$ is an HD user $j\neq k$) \\
10.\quad \quad \quad In sub-channel $a_l$ solve the problem  \textbf{\textit{P}3}  \\
11.\quad \quad \textbf{end for}
\\
12.\quad  \textbf{end for}
\\
13.\quad Using the obtained optimal powers, find the best pair $(k^*,j^*)$ in the \\
 \quad \quad sub-channel $a_l^*$ that  has the largest value of $L$
\\
14.\quad  $S_{j^*,u} \leftarrow [{S_{j^*,u},a_l}]$   ,   $S_{k^*,d} \leftarrow [{S_{k^*,d},a_l}]$ \\
15.\quad \textbf{if} $p_{k^*}\neq 0$ then $d_0(n)=d_0(n)+1$; \\
16.\quad \textbf{if} $p_{j^*}\neq 0$ then $d_{j^*}(n)=d_{j^*}(n)+1$; \\
17.\textbf{end for} \\
 \hline
 \end{tabular}
\end{table}
The complexity of finding the best user in each sub-channel is $O(K)$ and for $N$ sub-channels is $O(KN)$. Similarly, the complexity of finding the best pair in each sub-channel is $O(K^2)$ and doing so for $N$ sub-channels requires $O(NK^2)$ operations. Since the complexity of sorting $N$ values is $O(N \log N)$, then the overall computational complexity of the proposed sub-channel allocation algorithm is $O(N \log N+NK^2)$.
\section{Power Allocation}        \label{power_allocation}
The power allocation problem, denoted by \textbf{\textit{P}4}, can be formulated as follows 
\begin{align*}
\textbf{\textit{P}4}: & \operatorname*{maximize}_{p_{k,d},p_{j,u}} \qquad &R_d+R_u\\
& \text{subject to} \qquad &\text{(4)-(6)}
\end{align*}
Due to the interference terms, the power allocation problem is non-convex. Here, we use the ``difference of two concave 
functions/sets'' (DC) programming technique \cite{tuy2013convex} to convexify this problem. In this procedure, the non-concave 
objective function is expressed as the difference of two concave 
functions, and the discounted term is approximated by its first order 
Taylor series. Hence, the objective becomes concave and can 
be maximized by known convex optimization methods. 
This procedure runs iteratively, and after each iteration the 
optimal solution serves as an initial point for the next iteration 
until the improvement diminishes in iterations. In \cite{kha2012fast}, the DC approach is used to formulate optimized power allocation in a multiuser interference channel, and in \cite{mili2016energy}, the DC optimization method is used to optimize the energy efficiency of an OFDMA device to device network. Here, we rewrite the objective function of \textbf{\textit{P}4} in DC form as follows
\vspace{-2mm}
\begin{align*}
\operatorname*{max}_{\mathbf{p}} \quad 
f(\mathbf{p})-h(\mathbf{p}) 
\end{align*}

 \vspace{-9mm}

\begin{align*}
f(\mathbf{p})&=\sum_{k=1}^{K}\sum_{n\in S_{k,d}}w_k\log(N_k+I_{k,j}(n)p_{j,u}(n)+ g_k(n)p_{k,d}(n)) \\
&+\sum_{j=1}^{K}\sum_{n\in S_{j,u}}v_j\log(N_0+\beta p_{k,d}(n)+ g_j(n)p_{j,u}(n)) 
\end{align*}
\vspace{-7mm}
\begin{align*}
h(\mathbf{p})&=\sum_{k=1}^{K}\sum_{n\in S_{k,d}}w_k\log(N_k+I_{k,j}(n)p_{j,u}(n))\\
&+\sum_{j=1}^{K}\sum_{n\in S_{j,u}}v_j\log(N_0+\beta p_{k,d}(n)) 
\end{align*}

\normalsize
where  
\begin{align*}
\mathbf{p}=[p_{k_1,d}(1),...,p_{k_N,d}(N),p_{j_1,u}(1),,...,p_{j_N,u}(N)]^T 
\end{align*} is the downlink and uplink transmitted power vector, and $k_i$ and $j_i$ denote the uplink and downlink users that are selected for the $i$th sub-channel after the sub-channel allocation phase. Now, the objective $f(\mathbf{p})-h(\mathbf{p})$ is a DC function. To write the Taylor series of the discounted function $h(\mathbf{p}$), we need its gradient, that can be easily derived as follows.
\small{
\begin{align*}
&\nabla h(\mathbf{p})= \bigg[\frac{u_{j_1}\beta}{\ln(2)}\frac{1}{N_0+\beta p_{k_1,d}(1))},. . . ,\frac{u_{j_N}\beta}{\ln(2)}\frac{N}{N_0+\beta p_{k_N,d}(N))}, \\ &\frac{w_{k_1}I_{k_1,j_1}(1)}{\ln(2)}\frac{1}{N_{k_1}+I_{k_1,j_1}(1)p_{j_1,u}(1))},..., \\
&\frac{w_{k_N}I_{k_N,j_N}(N)}{\ln(2)}\frac{1}{N_{k_N}+I_{k_N,j_N}(N)p_{j_N,u}(N))} \bigg] ^T 
\end{align*}
}

\normalsize
To make the problem convex, $h(\mathbf{p})$ is approximated with its first order approximation $h(\mathbf{p}^{(t)})+\nabla h^T(\mathbf{p}^{(t)})(\mathbf{p}-\mathbf{p}^{(t)})$ at point $\mathbf{p}^{(t)}$. We start from a feasible $\mathbf{p}^{(0)}$ at the first iteration, and $\mathbf{p}^{(t+1)}$ at the $t$th iteration is generated as the optimal solution of the following convex program
\begin{align*}
\mathbf{p}^{(t+1)}=\operatorname*{arg max}_{\mathbf{p}} \quad &f(\mathbf{p})- h(\mathbf{p}^{(t)})-\nabla h^T(\mathbf{p}^{(t)})(\mathbf{p}-\mathbf{p}^{(t)}) \\ 
&\text{subject to } \ (4)-(6)
\end{align*}
Since $h(\mathbf{p})$ is a concave function, its gradient is also its super gradient so we have
\begin{align*}
h(\mathbf{p})\leq h(\mathbf{p}^{(t)})+\nabla h^T(\mathbf{p}^{(t)})(\mathbf{p}-\mathbf{p}^{(t)}), \quad  \quad \forall \mathbf{p} 
\end{align*}
and we can deduce
\begin{align*}
h(\mathbf{p}^{(t+1)})\leq h(\mathbf{p}^{(t)})+\nabla h^T(\mathbf{p}^{(t)})(\mathbf{p}^{(t+1)}-\mathbf{p}^{(t)}). 
\end{align*}
Then it can be proved that in each iteration the solution of  problem \textbf{\textit{P}4} is improved as follows
\begin{align*}
&f(\mathbf{p}^{(t+1)})- h(\mathbf{p}^{(t+1)}) \geq \\ &f(\mathbf{p}^{(t+1)})- h(\mathbf{p}^{(t)})-\nabla h^T(\mathbf{p}^{(t)})(\mathbf{p}^{(t+1)}-\mathbf{p}^{(t)}) \\
&=\operatorname*{max}_{\mathbf{p}} \quad f(\mathbf{p})- h(\mathbf{p}^{(t)})-\nabla h^T(\mathbf{p}^{(t)})(\mathbf{p}-\mathbf{p}^{(t)}) \\
& \geq f(\mathbf{p}^{(t)})- h(\mathbf{p}^{(t)})-\nabla h^T(\mathbf{p}^{(t)})(\mathbf{p}^{(t)}-\mathbf{p}^{(t)}) \\
&=f(\mathbf{p}^{(t)})- h(\mathbf{p}^{(t)}).
\end{align*}
According to the above equations, the objective value after each iteration is either unchanged or improved and since the constraint set is compact it can be concluded that the above DC approach converges to a local maximum.  

\section{Analyzing Self-Interference Cancellation Coefficient Threshold} \label{Coefficient_Threshold}

In \cite{peyman}, through simulations it has been observed that in a network that contains an imperfect FD BS and some imperfect FD and HD users, when the self-interference  cancellation coefficient is larger than a specified threshold, there is no difference between the throughput of an all HD user network and an all FD user network. Here we wish to  analyze this threshold. \\ 
Recall that in  FD networks there are four possible types of connections in a given sub-channel:
\begin{enumerate}
  \item HD downlink
  \item HD uplink
  \item Joint downlink and uplink for two distinct users over an FD BS
  \item A full-duplex bidirectional connection between an FD user and an FD BS
\end{enumerate} 

Employing an FD user in a cell with an FD capable BS can increase the throughput when the 4th case is more appealing than the other cases in at least one sub-channel. Considering sub-channel $n$, we assume that user $d$ is the best downlink user, user $u$ is  the best uplink user, users $a$ and $b$ are the best downlink and uplink pair, and user $f$ is the best FD node for FD communication with the BS. Here the best user, is the user who gives the highest weighted rate with the same power than the rest. The rate of the four previous cases are presented below (we drop the sub-channel index $n$ for simplicity)
\begin{equation}
R_d = w_d \log \bigg(1+ \frac{g_d p_{Bd}}{N_d}\bigg)  
\end{equation}

\begin{equation}
R_u = v_u \log \bigg(1+ \frac{g_u p_{uB}}{N_0}\bigg) 
\end{equation}

\begin{align}
R_{du} = & w_{a} \log \bigg(1+ \frac{g_{a} p_{Ba}}{N_a+g_{ba}p_{bB}}\bigg) +   v_b \log \bigg(1+ \frac{g_b p_{bB}}{N_0+\beta p_{Ba}}\bigg)
\end{align}

\begin{equation}
R_{f} =  w_{f} \log \bigg(1+ \frac{g_{f} p_{Bf}}{N_f+\beta p_{fB}}\bigg) + v_f \log \bigg(1+ \frac{g_f p_{fB}}{N_0+\beta p_{Bf}}\bigg)
\end{equation}

where, $p_{Bx}$ and $p_{yB}$ are the transmission powers form BS to user $x$ and from user $y$ to the BS, respectively. The other variables were introduced in  Table I. In short, using an FD user in the network could be beneficial when these conditions hold in at least one sub-channel:

\begin{equation} \label{condition1}
R_{f}>R_{d} 
\end{equation}
\begin{equation} \label{condition2}
R_{f}>R_{u} 
\end{equation}
\begin{equation} \label{condition3}
R_{f}>R_{du}.			
\end{equation}

Since we wish to focus on parameter $\beta$, and to avoid dealing with other parameters,  we introduce some simplifications. First, we assume the sum rate case where, $w_d=v_u=w_{a}=w_{f}=v_b=v_f=1$. Second, we assume that the noise powers at the BS and at the users are the same. Third, we assume that the transmission power from the BS to all users is the same and is equal to the average BS power, i.e., $p_{Bd}=p_{Ba}=p_{Bb}=p_{Bf}=\frac{P_0}{N}=P_{BS}$. Fourth, we assume  that the transmission powers from different users to the BS are the same and are equal to the average user power, i.e., $p_{uB}=p_{bB}=p_{fB}=\frac{KP_K}{N}=P_{user}$. Fifth, the channel gains $g_d$, $g_u$, $g_f$, $g_a$, $g_b$, $g_{ba}$ are random variables and in the sum rate case the best downlink user, the best uplink user and the best FD user are all the same $g_d=g_u=g_f=g_{max}$, because the channel is reciprocal and the user with maximum channel gain is selected for all of these three cases.
If we assume that the number of users in the network is $K$, then random variable $g_{max}$  can be defined as:
$g_{max}=\text{max}\left\{g_1,g_2,...,g_K\right\}$, where $g_i$ is the random channel gain between the BS and the user $i$, that itself is a multiplication of an exponential random variable, $e_i$, with unit power and a path loss random variable $l_i$ that depends on the path loss model and the distance $d_i$ between the BS and the $i$th user whose pdf is shown by $f_{D_i}(d_i) = \frac{2d_i}{R_{cell}^2}$ (for $0<d_i<R_{cell}$). We assume $g_{ab}$ is a random channel gain between  two users $a$ and $b$ which are distributed uniformly in a circle with radius $R_{cell}$.We also assume that $g_a$ and $g_b$ are the maximum and the second maximum channel gain between $K$ users.

Due to the randomness of the channel gains, $\beta$ itself is a random variable and here we wish to derive its distribution. According to   conditions (\ref{condition1}) - (\ref{condition3}), we have:
\begin{enumerate}
\item The FD rate should be bigger than the HD downlink rate, so we have:
\begin{equation*}
\log (1+ \frac{g_{f} P_{BS}}{N_0+\beta P_{user}}) +  \log (1+ \frac{g_f P_{user}}{N_0+\beta P_{BS}}) > \log (1+ \frac{g_d P_{BS}}{N_0}) 
\end{equation*}
After some manipulations, this is reduced to the  inequality $a_1\beta^2 + b_1\beta - c_1<0$, where: 
\begin{align}
a_1&=g_d P^2_{BS} P_{user} \\
b_1&=N_0 g_f P_{user} (P_{BS}-P_{user}) \\
c_1&=N_0^2g_f P_{user}+N_0 g^2_fP_{BS}P_{user} 
\end{align}  

which holds for  $0 <\beta<\frac{-b_1+\sqrt{b_1^2+4a_1c_1}}{2a_1}$. Therefore, due to condition (\ref{condition1})  $\beta_1=\frac{-b_1+\sqrt{b_1^2+4a_1c_1}}{2a_1}$.
\item By writing the condition (\ref{condition2}) and doing the same procedure as the previous part we arrive at  inequality $a_2\beta^2 + b_2\beta - c_2<0$ 
where:      
\begin{align}
a_2&=g_u P^2_{user} P_{BS} \\
b_2&=N_0 g_u P_{BS} (P_{user}-P_{BS}) \\
c_2&=N_0^2g_f P_{BS}+N_0 g^2_fP_{BS}P_{user}
\end{align}   
which holds for  $0 <\beta<\frac{-b_2+\sqrt{b_2^2+4a_2c_2}}{2a_2}$.  Therefore,  according to  (\ref{condition2})  $\beta_2=\frac{-b_2+\sqrt{b_2^2+4a_2c_2}}{2a_2}$. 
\item By writing the condition (\ref{condition3}) and doing the same procedure as the previous parts we arrive at the  inequality $a_3\beta^3 + b_3\beta^2 + c_3\beta + d_3 <0$, where: 

\begin{align}
&a_3=P_{BS}^3 P_{user} g_a  \\
\nonumber &b_3=P_{user}^2 P_{BS} g_b N_0 + P_{user}^3 P_{BS} g_b g_{ab} + P_{user}^2  P_{BS}^2 g_b g_a + \\
\nonumber &P_{BS}^2 P_{user}  g_a N_0 + P_{BS}^3 g_a N_0 - P_{user}^2 P_{BS} g_f N_0 - \\
 & P_{BS}^3 g_f N_0 - g_f P_{BS}^3  P_{user} g_{ab} - P_{user}^3 P_{BS} g_f g_{ab}\\
 \nonumber &c_3 = P_{user}^2 P_BS g_b N_0 g_a +  P_{user}^3 g_{ab} N_0 g_b + P_{user}^2 N_0^2 g_b + \\
 \nonumber & P_{user} P_{BS} N_0^2 g_a + P_{BS}^2 P_{user} g_b N_0 g_a + P_{user}^2 P_{BS} g_b N_0 g_{ab} + \\
\nonumber &N_0^2 P_{BS} g_b P_{user} + 2  N_0^2 P_{BS}^2 g_a - P_{user}^3 g_f  g_{ab} N_0 - \\
\nonumber &P_{BS}^2 g_f g_{ab}  N_0 P_{user} - P_{BS}^2 g_f^2  N_0 P_{user} - N_0^2 g_f P_{BS} P_{user} - \\
\nonumber &2 P_{BS}^2 N_0^2 g_f - P_{user}^2 N_0^2 g_f - P_{BS}^2 P_{user} g_f g_{ab} N_0 - \\
& P_{BS} g_{ab}  P_{user}^2 N_0 g_{f} - g_f^2 P_{BS}^2 P_{user}^2 g_{ab} \\
\nonumber &d_3=P_{BS} N_0^3  g_a + P_{user} N_0^3 g_b + P_{user}^2   N_0^2  g_b g_{ab} + \\
\nonumber &P_{user}P_{BS} N_0^2  g_b g_a -  P_{user}^2   N_0  g_f^2 g_{ab} P_{BS} - P_{user}^2 g_f N_0^2 g_{ab} - \\
\nonumber  &P_{user} g_f N_0^2) g_{ab} P_{BS} - P_{user} g_f^2 N_0^2 P_{BS} - P_{user} g_f N_0^3  - \\
& P_{BS} g_f N_0^3
\end{align}

\normalsize
The above cubic function has the following three roots:
\begin{align}
x_1=& A + B -  \frac{b}{3}  \\
x_2=& \frac{-1}{2}(A+B) + \frac{i \sqrt{3}}{2}(A-B) -  \frac{b}{3} \\
x_3=& \frac{-1}{2}(A+B) - \frac{i \sqrt{3}}{2}(A-B) -  \frac{b}{3}
\end{align} 

where 
\begin{align}
A=& \sqrt[3]{\frac{-q}{2}+ \sqrt{\frac{q^2}{4}+\frac{p^3}{27}}}  \\
B=& \sqrt[3]{\frac{-q}{2}- \sqrt{\frac{q^2}{4}+\frac{p^3}{27}}}  \\
p=& \frac{-b^2}{3}+c \\
q=& \frac{2b^3}{27}-\frac{bc}{3} + d 
\end{align} 
 
\begin{equation}
b=\frac{b_3}{a_3}, \quad c=\frac{c_3}{a_3}, \quad d=\frac{d_3}{a_3}
\end{equation}
 
It is evident that  $a_3\geq 0$. Also, it can be shown that the value of $d_3$ is always negative, therefore, it is deduced that  the  cubic function has at least one positive real root. Therefore, due to the condition (\ref{condition3}), $\beta_3= \min \left\lbrace  R(x_1),R(x_2),R(x_3) \right\rbrace $ is the smallest positive real root of this cubic function, where $R(x)$ is given by:
\vspace{-3mm}
\[ R(x) = \left\{
\begin{array}{lr}
x & \mbox{if } x\mbox{ is real and positive}\\
\infty & \mbox{Otherwise}.
\end{array}
\right.
\]
Finally, we arrive at the following proposition:

\begin{prop}
For a wireless cell with an FD BS and users with imperfect SI cancellation factor $\beta$, FD operation is advantageous from the perspective of the network throughput performance if $\beta < \beta_{Threshold}=\min \left\lbrace \beta_1,\beta_2,\beta_3 \right\rbrace $.
\end{prop}

\end{enumerate} 
In Section \ref{simulation_results}, we will compare the outcome of this analysis with simulation results.

\section{Two-Tier Heterogeneous Full Duplex Network} \label{Two-tier}

In this section, we consider a two-tier heterogeneous full-duplex OFDMA network. This system includes a macrocell FD BS and multiple femto cell FD BSs along with their associated HD and FD users. Our goal is to maximize the  uplink and downlik weighted sum rate of femto cell users while provisioning for the macrocell user's uplink and downlink data rate. 
Assume that the numbers of femto cells and  available sub-channels are  $M_f$ and $N$, respectively, and the number of users related to the $m$th BS is $K_m$. We denote the set of all BSs as $\Omega= \lbrace 0,1,2,...,M_f \rbrace$, where the macro BS is indexed by 0. The variables used in the following equations are summarized in Table II.

The downlink rate in cell $m$ is given by:

\begin{align}
&R_{d,m}=\sum_{k=1}^{K_m}\sum_{n\in S_{k,d,m}}w_{k,m}\log \bigg(1+ \nonumber \\
 & \frac{g_{k,m}(n)p_{k,d,m}(n)}{N_{k,m}+\sum_{j=1}^{K_m} I_{k,j,m}(n)p_{j,u,m}(n)+ DIC_{m,k}(n) + UIC_{m,k}(n) }\bigg)
\end{align}

where $DIC_{m,k}(n)$ and $UIC_{m,k}(n)$ are the downlink and uplink inter-cell interference on sub-channel $n$ in the $m$th cell for user $k$, i.e.: 
\begin{equation}
DIC_{m,k}(n)=\sum_{m^{'} \in \Omega \text{\textbackslash}  \{m\} } g_{k,m,m^{'}}(n) p_{d,m^{'}}(n)
\end{equation}
\begin{equation}
UIC_{m,k}(n)=\sum_{m^{'} \in \Omega \text{\textbackslash}  \{m\} } \sum_{j=1}^{K_{m^{'}}} g_{k,m,j,{m^{'}}}(n) p_{j,u,m^{'}}(n)
\end{equation}

Similarly the uplink rate in cell $m$ is given by:
\begin{align}
&R_{u,m}=\sum_{j=1}^{K_m}\sum_{n\in S_{j,u,m}}v_j\log \bigg(1+ \nonumber \\
 & \frac{g_{j,m}(n)p_{j,u,m}(n)}{N_m+ \beta \sum_{k=1}^{K_m} p_{k,d,m}(n)+ DIC_m(n) + UIC_m(n) } \bigg)
\end{align}
where $DIC_{m}(n)$ and $UIC_{m}(n)$ are the downlink and uplink inter-cell interference on sub-channel $n$ at the $m$th BS, i.e.:
\begin{equation}
DIC_{m}(n)=\sum_{m^{'} \in \Omega \text{\textbackslash}  \{m\} } g_{m,m^{'}}(n) p_{d,m^{'}}(n)
\end{equation}                                                      
\begin{equation}
UIC_{m}(n)=\sum_{m^{'}\in \Omega \text{\textbackslash}  \{m\} } \sum_{j=1}^{K_{m^{'}}}g_{j,{m^{'}},m}(n) p_{j,u,m^{'}}(n)
\end{equation}
\begin{table}[t!] \label{table:Heterogeneous Problem}
\centering
 \caption{Key Variables For The Heterogeneous Network}
 \begin{tabular}{|c | m{38em}|}

 \hline
 $w_{k,m}$ & weight assigned to the downlink  of user $k$ in the $m$th cell  \\
 \hline
 $v_{j,m}$ & weight assigned to the uplink of user $j$ in the $m$th cell \\
 \hline
  $p_{k,d,m}(n)$ &  downlink transmission power from the BS to user $k$ on sub-channel $n$ in the $m$th cell \\
 \hline
  $p_{j,u,m}(n)$ & uplink transmission power from user $j$ to the BS on sub-channel $n$ in the $m$th cell\\
 \hline
 $p_{d,m^{'}}(n)$ & downlink transmission power from the $m^{'}$th  BS on sub-channel $n$ \\
 \hline
 $g_{k,m,m^{'}}(n)$ & channel gain between user $k$ in the $m$th cell and the $m^{'}$th BS on sub-channel $n$ \\
 \hline
 $g_{k,m,j,{m^{'}}}(n)$ & channel gain between  user $k$ in the $m$th cell and user $j$ in the $m^{'}$th cell on sub-channel $n$ \\
  \hline
  $g_{k,m}(n)$ & channel gain between the BS and user $k$ on sub-channel $n$ in the $m$th cell\\
  \hline
  $g_{m,m^{'}}(n)$ & channel gain between the $m$th BS and the $m^{'}$th BS on sub-channel $n$\\
 \hline
  $ N_{k,m}$  & Gaussian noise variance at the receiver of user $k$ in the $m$th cell \\
 \hline
  $N_m$ & Gaussian noise variance at the $m$th base-station receiver\\
 \hline
  $S_{k,d,m}$ & set of sub-channels allocated to user $k$ for downlink transmission in the $m$th cell\\
 \hline
  $S_{j,u,m}$ & set of sub-channels allocated to user $j$ for uplink transmission in the $m$th cell\\
 \hline
  $\beta$ & self-interference cancellation coefficient \\
 \hline
 $I_{k,j,m}(n)$ & equal to $\beta$ when $j=k$, and to $g_{k,m,j,{m^{}}}(n)$ otherwise\\
 \hline
 $P^{BS}_m$ & maximum available transmit power at the $m$th BS\\
 \hline
 $P_{k,m}$ & maximum available transmit power at user $k$ in the $m$th cell\\
 \hline
 $R_{mind}$ & minimum required downlink rate for the macrocell\\
 \hline
 $R_{minu}$ & minimum required uplink rate for the macrocell\\
 \hline
 \end{tabular}
\end{table}

The optimization problem for the heterogeneous network can be formulated as follows:
\begin{equation}\label{eq:21}
\textbf{\textit{P}}_{\textbf{Het}}: \operatorname*{Maximize}_{p_{k,d,m},p_{j,u,m},S_{j,u,m},S_{k,d,m}} \qquad  \sum_{m=1}^{M_f} R_{d,m}+R_{u,m}
\end{equation}
\begin{equation}\label{eq:22}
\text{Subject to }   \sum_{k=1}^{K_m}\sum_{n\in S_{k,d,m}}p_{k,d,m}(n) \leq P^{BS}_m \quad \forall m
\end{equation}
\begin{equation}\label{eq:23}
\sum_{n\in S_{j,u,m}}p_{j,u,m}(n) \leq P_{j,m} \quad \forall j,m
\end{equation}
\begin{equation}\label{eq:24}
R_{d,0} \geq R_{mind}, \ R_{u,0} \geq R_{minu}
\end{equation}
\begin{equation}\label{eq:26}
p_{j,u,m}(n),p_{k,d,m}(n)\geq 0 \quad \forall j,k,n,m
\end{equation}
\begin{equation}\label{eq:27}
S_{i,d,m}\cap S_{j,d,m}=\phi, \ S_{i,u,m}\cap S_{j,u,m}=\phi \quad \forall i\neq j,   \forall m
\end{equation}
\begin{equation}\label{eq:29}
\cup_{j=1}^{K_m}  \ S_{j,u,m} \subseteq \{1,2,...,N\}, \ \cup_{k=1}^{K_m} \ S_{k,d,m} \subseteq \{1,2,...,N\} \quad \forall m
\end{equation}
\begin{equation}\label{eq:31}
S_{k,u,m} \cap S_{k,d,m}=\phi \quad \text{if user} \ k \  \text{is} \ \text{HD} \quad \forall m
\end{equation}
where (\ref{eq:22}) and (\ref{eq:23}) indicate the power constraint on the BSs and the users, respectively; (\ref{eq:24}) is the minimum downlink and uplink rate constraints for the macrocell users. The constraint (\ref{eq:26}) shows the non-negativity of transmission powers; (\ref{eq:27}) comes from the fact that a sub-channel cannot be allocated to two distinct users simultaneously; (\ref{eq:29}) indicates that we have no more than $N$ sub-channels, and the last constraint accounts for the half-duplex nature of the HD users.  
%\begin{align}
%\textbf{\textit{P}1}: \operatorname*{maximize}_{p_{k,d,m},p_{j,u,m},S_{j,u,m},S_{k,d,m}} \qquad  \sum_{m=1}^{M_f} R_{d,m}+R_{u,m}\\
%\text{S.T }   \sum_{k=1}^{K_m}\sum_{n\in S_{k,d,m}}p_{k,d,m}(n) \leq P^{BS}_m \quad \forall m \\
%\sum_{n\in S_{j,u,m}}p_{j,u,m}(n) \leq P_{j,m} \quad \forall j,m\\
%R_{d,0} \geq R_{mind} \\
%R_{u,0} \geq R_{minu} \\
%p_{j,u,m}(n),p_{k,d,m}(n)\geq 0 \quad \forall j,k,n,m\\
% S_{i,d,m}\cap S_{j,d,m}=\phi \quad \forall i\neq j\\
% S_{i,u,m}\cap S_{j,u,m}=\phi \quad \forall i\neq j\\
%\cup_{j=1}^{K_m}  \ S_{j,u,m} \subseteq \{1,2,...,N\} \\
%\cup_{k=1}^{K_m} \ S_{k,d,m} \subseteq \{1,2,...,N\} \\
%S_{k,u,m} \cap S_{k,d,m}=\phi \quad \text{if user} \ k \  \text{is} \ \text{HD}
%\end{align}
\begin{figure}
 \centering
  \includegraphics[width=9cm, height=6.5cm]{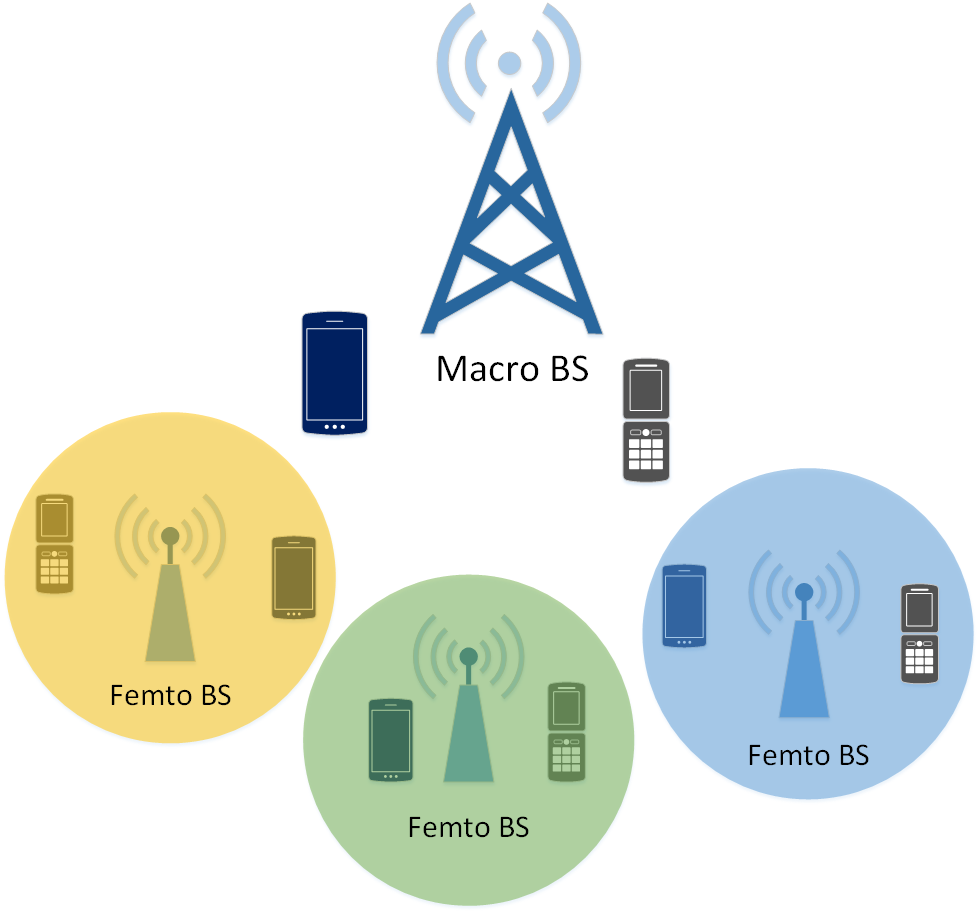}
  \caption{\small{A two-tier heterogeneous OFDMA full-duplex network that contains an imperfect full-duplex macro base-station and multiple femto cell BSs and their associated users.}}
\vspace{-8mm}
\end{figure}  
To address this problem, we propose a scheme which optimizes power allocation and sub-channel assignment in an iterative manner. At the beginning of each iteration $t$, we find the
proper sub-channel assignment $\mathbf{S}[t]$  for the power allocation
obtained from the last iteration $\mathbf{P}[t-1]$. Then for this $\mathbf{S}[t]$, we find the
optimal power allocation. We repeat the process in all subsequent iterations until no further noticeable improvement is observed, i.e.: 
\begin{align*}
\mathbf{S}[0]\longrightarrow \mathbf{P}[0] \longrightarrow  \cdots\cdots \longrightarrow \mathbf{S}[t]\longrightarrow \mathbf{P}[t] 
\end{align*}
 
At the first iteration for sub-channel allocation, in each femto cell, sub-channels are allocated based on  Algorithm 1 in Section \ref{subchannel allocation} without considering the inter-cell interference. At the macro cell, since uplink-downlink rate constraints are to be satisfied, additional  considerations are required. Algorithm 2 presents a solution for the rate-constrainted sub-channel allocation. Depending on whether $R_{mind}$ or $R_{minu}$ is larger, the algorithm allocates a downlink or uplink sub-channel and then  the estimated resulting rate of the new added sub-channel is subtracted from the required minimum rate. This procedure is repeated until  both  $R_{mind}$ and $R_{minu}$ become equal or less than zero. In this case, a  sufficient number of sub-channels has been allocated to uplink and downlink users in order to satisfy the rate constraints. After that, the  algorithm switches to the one without the rate constraint.
\small{
\begin{table}[h!]
\centering
 \begin{tabular}{m{30em}}
 \hline
 \textbf{Algorithm 2: Sub-channel Allocation Algorithm with Rate Constraint} \\
 \hline\hline
1.\textbf{for} $n=1$ to $N$ \textbf{do} \\ 
2.\quad $\tilde{g}({n})=\operatorname*{max}_k g_k(n) $ \\
3.\textbf{end for} \\
4.Find a sub-channel permutation $\{ a_1,...,a_N \}$, $a_i\in\{1,...,N\}$, $a_i\neq a_j$ such that\\
\hspace{0.5cm} $\tilde{g}({a_1}) \geq \tilde{g}({a_2}) \geq . . .\geq \tilde{g}({a_N})$ \\
5. set $d_k(l)=1$  for $0\leq k \leq K$  and $1 \leq l \leq N$ \\
6.\textbf{for} $l=1$ to $N$ \textbf{do} \\
7.\quad Set $P_{max1}=\dfrac{P_0}{d_0(l)}$  and  $P_{max2}=\dfrac{P_k}{d_k(l)}      \forall k$  \\
8.\quad   \textbf{if}($(R_{mind} \geq 0$) and $(R_{mind} \geq R_{minu})$ )    \\
\quad  \quad  \textbf{begin} \\
9.\quad \quad  In sub-channel $a_l$ find the best downlink user $k^*$  \\
10.\quad  \quad $R_{mind} \longleftarrow R_{mind}- L(\frac{P_{0}}{d_{0}(l)},0)$\\
11.\quad  \quad  $S_{k^*,d} \leftarrow [{S_{k^*,d},a_l}]$ and  $d_0(n)=d_0(n)+1$; \\
\quad  \quad   \textbf{end} \\
12.\quad  \textbf{elseif}($(R_{minu} \geq 0)$ and  ($R_{minu} \geq R_{mind}$)) \\
\quad  \quad \textbf{begin} \\
13.\quad \quad  In sub-channel $a_l$ find the best uplink user $j^*$  \\
14.\quad  \quad $R_{minu} \longleftarrow R_{minu}- L(0,\frac{P_{j^*}}{d_{j^*}(l)} )$\\
15.\quad  \quad  $S_{j^*,u} \leftarrow [{S_{j^*,u},a_l}]$  and  $d_{j^*}(l)=d_{j^*}(l)+1$;  \\
\quad \quad  \textbf{end}\\
16. \quad  \textbf{else}  \\
\quad  \quad \textbf{begin} \\
17. \quad \textbf{for} $k=1$ to $K$ \textbf{do} \\
18. \quad \quad 	\textbf{for} $j=1$ to $K$ (if $k$ is an HD user $j\neq k$) \\
19.\quad \quad \quad In sub-channel $a_l$ solve the problem  \textbf{\textit{P}3}  \\
20.\quad \quad \textbf{end for}
\\
21.\quad  \textbf{end for}
\\
22.\quad Using the obtained optimal powers, find the best pair $(k^*,j^*)$ in the \\
 \quad \quad sub-channel $a_l^*$ that  has the largest value of $L$
\\
23.\quad  $S_{j^*,u} \leftarrow [{S_{j^*,u},a_l}]$   ,   $S_{k^*,d} \leftarrow [{S_{k^*,d},a_l}]$ \\
24.\quad \textbf{if} $p_{k^*}\neq 0$ then $d_0(n)=d_0(n)+1$;\\
25.\quad \textbf{if} $p_{j^*}\neq 0$ then $d_{j^*}(n)=d_{j^*}(n)+1$; \\
\quad  \quad \textbf{end}\\
26.\textbf{end for} \\
 \hline
 \end{tabular}
\end{table}
}\normalsize
After sub-channel allocation at the first iteration, we perform power allocation by using the DC approach as described  in Section \ref{power_allocation} in order to convexify the objective function and the rate constraints. Then, for the next iterations, we perform sub-channel allocation by considering the inter-cell interference. To choose the best pair in each cell we use   \textbf{Proposition 1} in  Section \ref{subchannel allocation}. For a heterogeneous network, where $N_k$ is replaced by  $N_k+DIC_{m,k}(n) + UIC_{m,k}(n)$ because in addition to Gaussian noise at the $k$th user we should take into account the uplink and downlnk interference from other cells. Similarly, because of the uplink and downlink interference at the BS we replace $N_0$ with $N_0+DIC_{m}(n) + UIC_{m}(n)$. The proof of convergence of power iterations is the same as in Section \ref{power_allocation}. For the sub-channel iterations, as many interfering nodes exist, the mathematical proof of convergence is intractable, but simulation results show that it converges to a local maximum.

\section{Simulation Results}
\label{simulation_results}
In this Section, we evaluate the proposed resource allocation scheme for OFDMA networks with half-duplex and imperfect full-duplex nodes.  We assume a time-slotted system, where nodes are uniformly distributed within a given cell radius. Table III presents the details of the indoor and outdoor simulation setup and channel models for the single cell network. In addition to the path loss, a Rayleigh block fading channel model with unit average power is considered. The channel gains remain constant in each time slot and vary independently from one time slot to the next.

\begin{table}
\centering
 \caption{simulation parameters}
 \begin{tabular}{|m{13em} | m{28em}|}

 \hline
 \textbf{PARAMETER} & \textbf{VALUE} \\
 \hline
 Maximum BS Power (outdoor) & $43$ dBm \\
 \hline
 Maximum BS Power (indoor) & $24$ dBm \\
 \hline
  Maximum UE Power $(P_k)$ & $23$ dBm \\
 \hline
  Thermal Noise Density & $-170$ dBm/Hz\\
 \hline
  Number of Sub-channels $N$ & $64$\\
 \hline
  Total Bandwidth & $10$ MHz \\
 \hline
  Sub-channel Bandwidth  & $150$ KHz \\
  \hline
  Cell Radius (outdoor) & $1$ km\\
   \hline
  Cell Radius (indoor) & $20$ m\\
 \hline
 Center Frequency & $2$ GHz\\
 \hline
 BS to UE Path Loss (outdoor) & urban Hata model with parameters                                  $h_m= 1.5$ m, $h_B= 30$ m \\   
 \hline
 UE to UE Path Loss (outdoor) & urban Hata model with parameters                                  $ h_m= 1.5$ m, $h_B= 1.5$ m \\
 \hline
  Path Loss Model (indoor) & ITU model for indoor attenuation with parameters N= $22$, $p_f(n)= 9$ \\   
 \hline
 \end{tabular}
\end{table}

For comparison of different cases in a single cell network, we consider six schemes: (i) An HD uplink system (HD-U), (ii) An HD downlink system (HD-D), (iii) a system that includes an FD BS and  HD users (FD-HD), (iv) a system that contains  an FD BS and FD users (FD-FD), (v) an upper bound which is the HD uplink rate plus the HD downlink rate; (vi) a Hybrid HD scheme (HHD), in which a hybrid HD BS could transmit data to downlink users and receive data from uplink users simultaneously in different sub-channels. 
For the HD-D case, each sub-channel is allocated to the user with the best weighted channel SNR, and multi-level water filling \cite{seong2006optimal} is applied  for power allocation. For the sub-channel assignment of the HD-U scheme the $SOA1 \: 4B \: 5A$ method presented in \cite{huang2009joint} is used, and for power allocation each user performs water filling in its dedicated sub-channels. In the HHD scheme, we use the proposed sub-channel allocation algorithm by changing the set $P_{opt}$ to:
 \begin{align*}
P_{opt}=\left\{(0,P_{max2}),(P_{max1},0)\right\}
\end{align*}
and perform multi-level water filling and water filling for the power allocation in the selected downlink and uplink sub-channels, respectively.

\begin{figure}[t!]
\begin{minipage}{\linewidth}
\includegraphics[width=\linewidth]{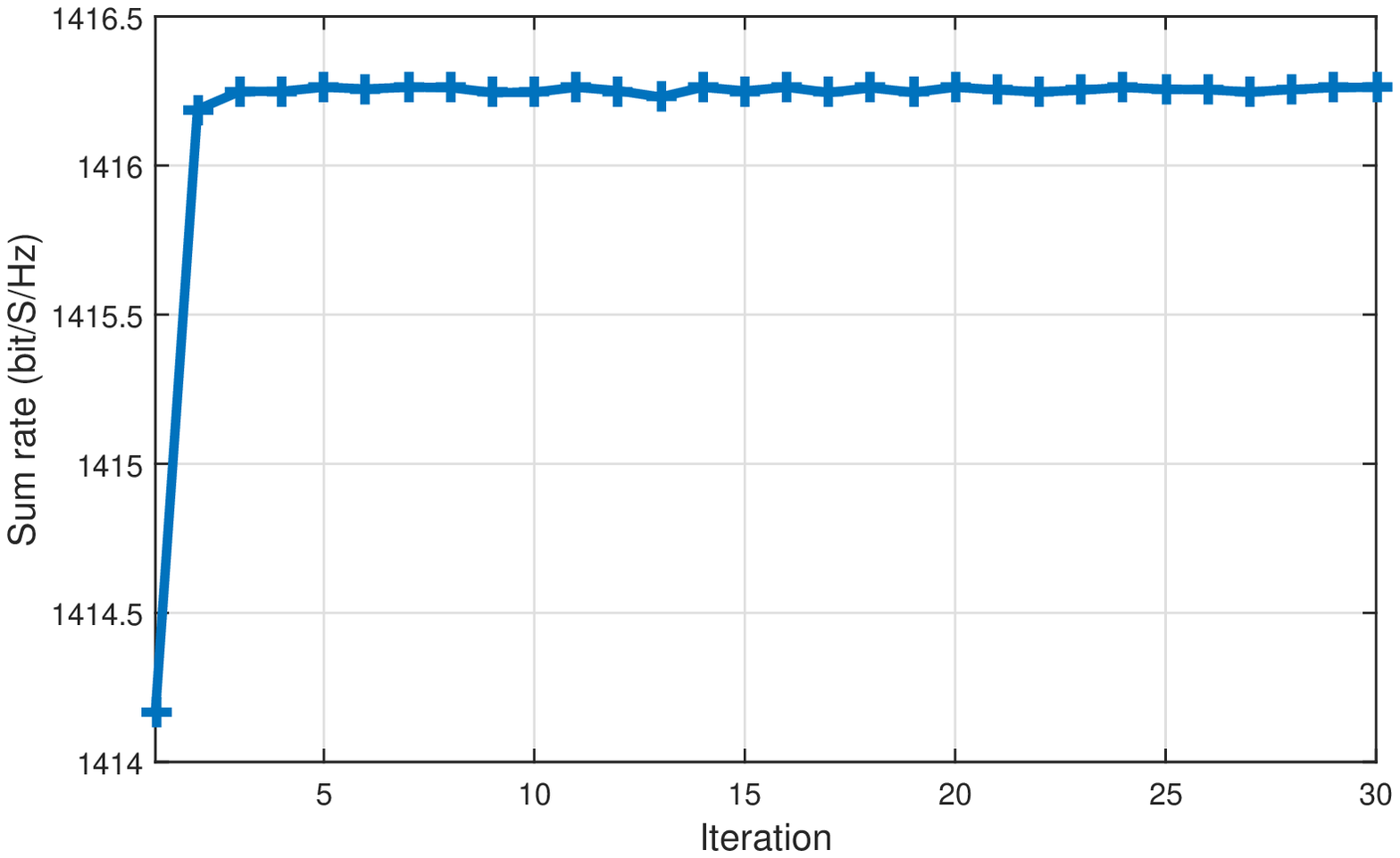}
\caption{\small{Convergence behavior of the proposed algorithm. Here $\beta=10^{-6}$ and other simulation parameters are the same as in the outdoor case.}}
\label{fig:convergence}
\end{minipage}
\hfill
\begin{minipage}{\linewidth}
\includegraphics[width=\linewidth]{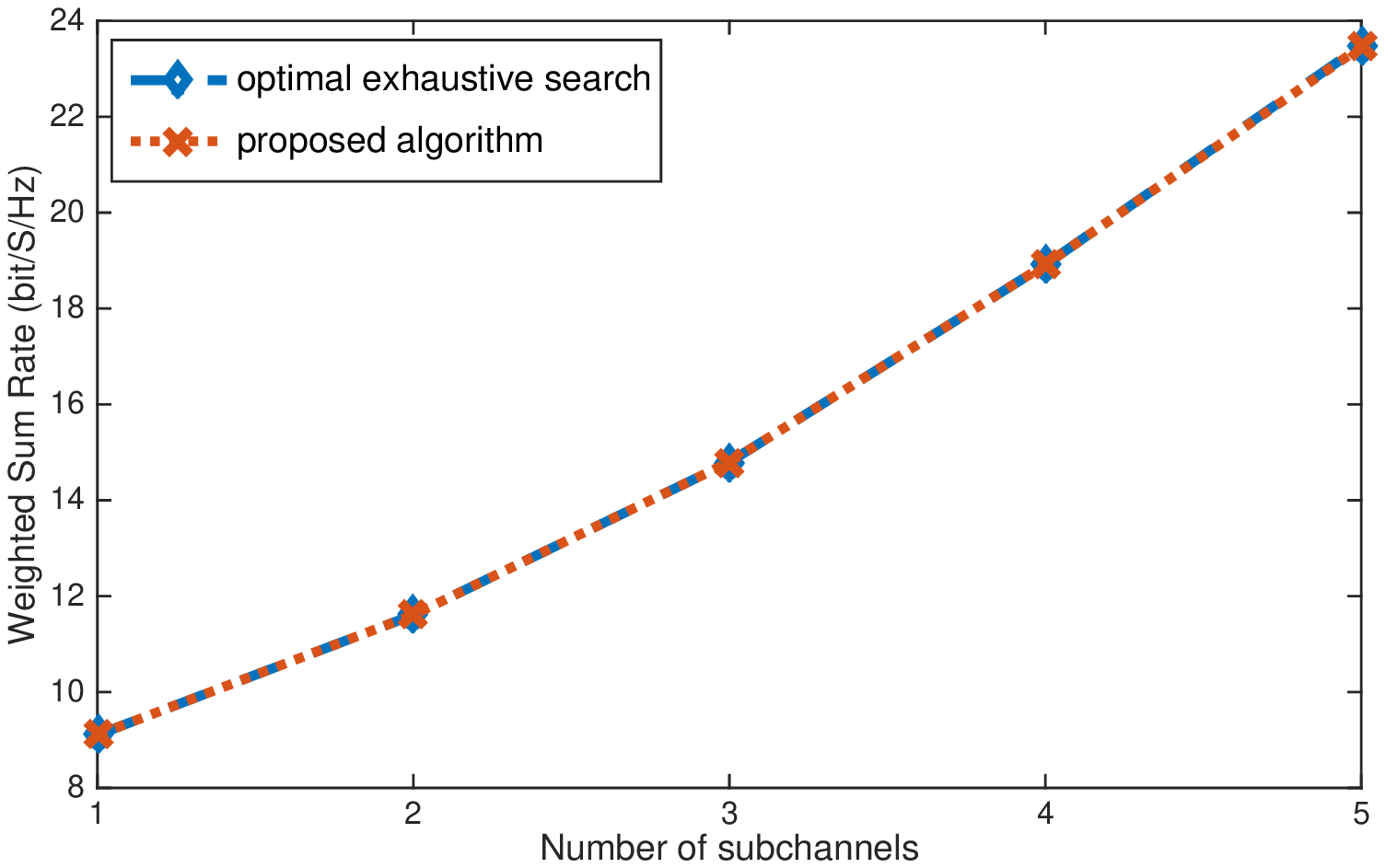}
\caption{\small{Performance comparison of the proposed algorithm with optimal exhaustive search in a small network.}}
\label{fig:Compare_Optimal}
\end{minipage}%
\vspace{-8mm}
\end{figure}

Fig. \ref{fig:convergence} illustrates the convergence of the proposed resource allocation scheme in a single cell OFDMA network with 10 HD nodes and 10 imperfect FD nodes. As can be seen, the sum-rate converges in just a few iterations.

Fig. \ref{fig:Compare_Optimal} compares the proposed algorithm with the optimal exhaustive search solution. Due to the high computational complexity of exhaustive search, only a small network with one HD and one FD user and a small number of sub-channels can be considered. Uplink and downlink weight vectors are assumed to be $\mathbf{u}=[1/3, 2/3]^{T}$ and $\mathbf{w}=[2/3, 1/3]^{T}$ respectively, and the SI cancellation coefficient is set to $\beta = -90$ dB. Simulation results show that, at least for small size networks, our proposed algorithm achieves the performance of the optimal exhaustive search.

%\begin{figure}[t!]
%  \includegraphics[width=\linewidth]{all_outdoor2.eps}
%  \caption{Performance comparison of six schemes in FD and HD networks in the outdoor scenario.}
%    \label{fig:all_outdoor}
%\end{figure}
%
%\begin{figure}[t!]
%  \includegraphics[width=\linewidth]{all_indoor2.eps}
%  \caption{Performance comparison of six schemes in FD and HD networks in the indoor scenario.}
%      \label{fig:all_indoor}
%\end{figure}

Fig. \ref{fig:all_outdoor} shows the sum-rate of the different schemes in the outdoor scenario with perfect SI cancellation ($\beta=0$). It can be seen that when the BS and all nodes are perfect FD devices the upper-bound could be attained, and when the nodes are HD but the BS is FD the sum-rate is still bigger than the cases with HD BS, but it can not reach the upper-bound because of inter-node interference.

Fig. \ref{fig:all_indoor} shows the sum-rate of the six presented schemes in an indoor scenario. If we compare the outdoor 
and indoor scenarios we find that using an FD BS in an outdoor environment has much larger gain  than  using it in an indoor case. This result is  intuitive because in the outdoor environment the distances between  nodes are larger, and hence the inter-node interference is smaller. As a result, the FD BS could work in FD mode in more sub-channels, which helps increase the network throughput more significantly.
\begin{figure}[t!]
\begin{minipage}{\linewidth}
\includegraphics[width=\linewidth]{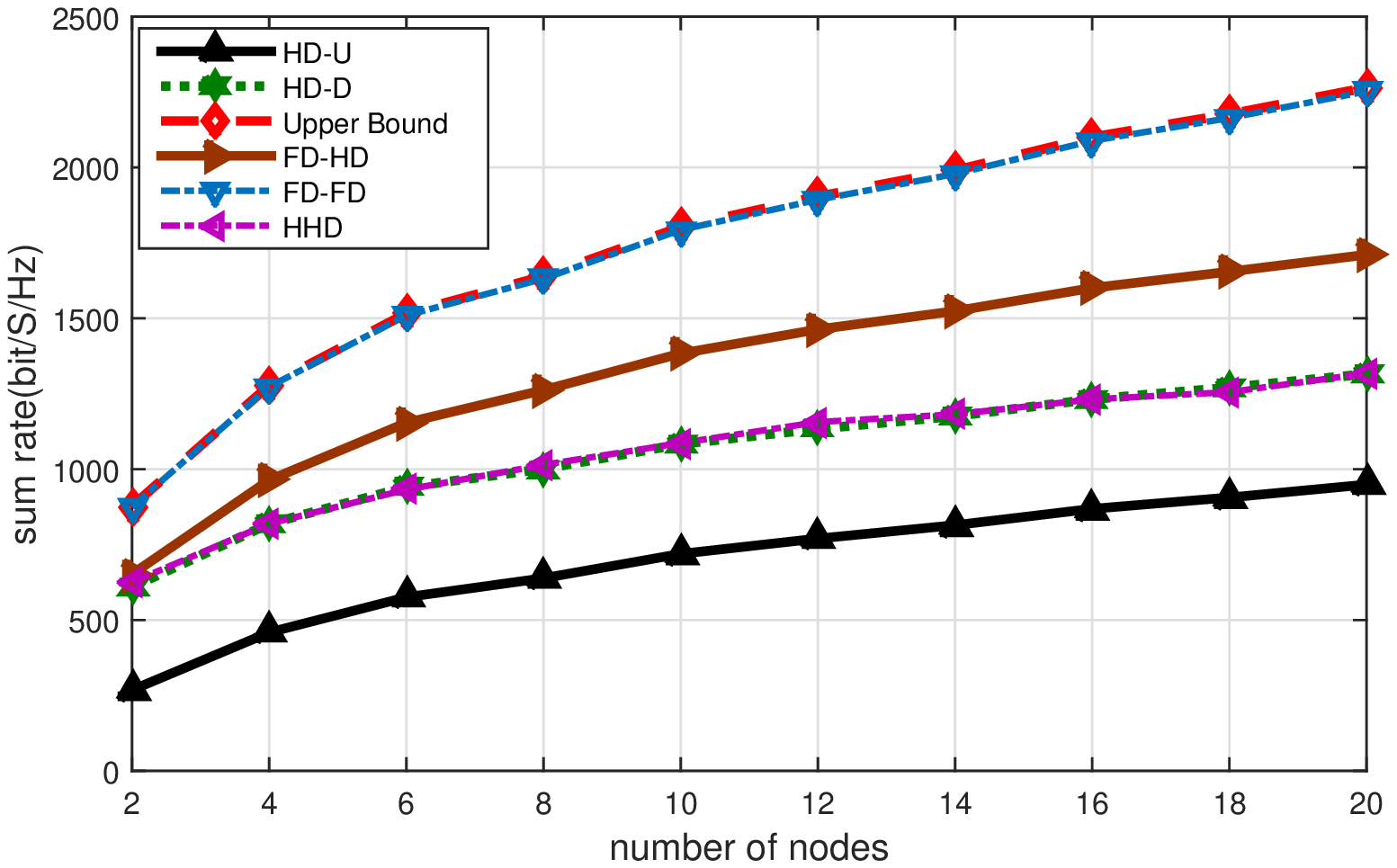}
\caption{\small{Performance comparison of six schemes in FD and HD networks in the outdoor scenario.}}
\label{fig:all_outdoor}
\end{minipage}
\hfill
\begin{minipage}{\linewidth}
\includegraphics[width=\linewidth]{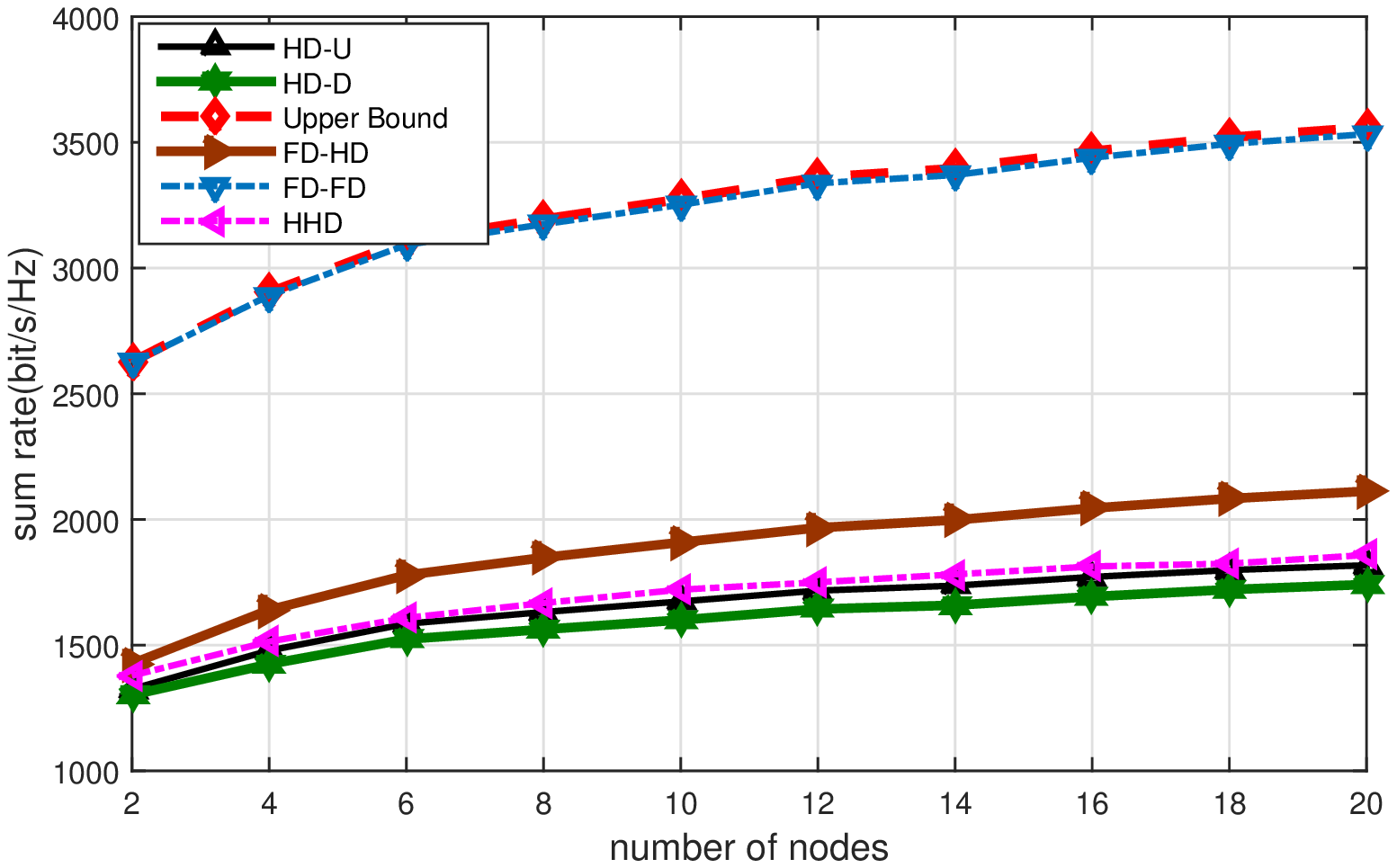}
\caption{\small{Performance comparison of six schemes in FD and HD networks in the indoor scenario.}}
\label{fig:all_indoor}
\end{minipage}%
\vspace{-8mm}
\end{figure}

\begin{figure}[t!]
\begin{minipage}{\linewidth}
\includegraphics[width=\linewidth]{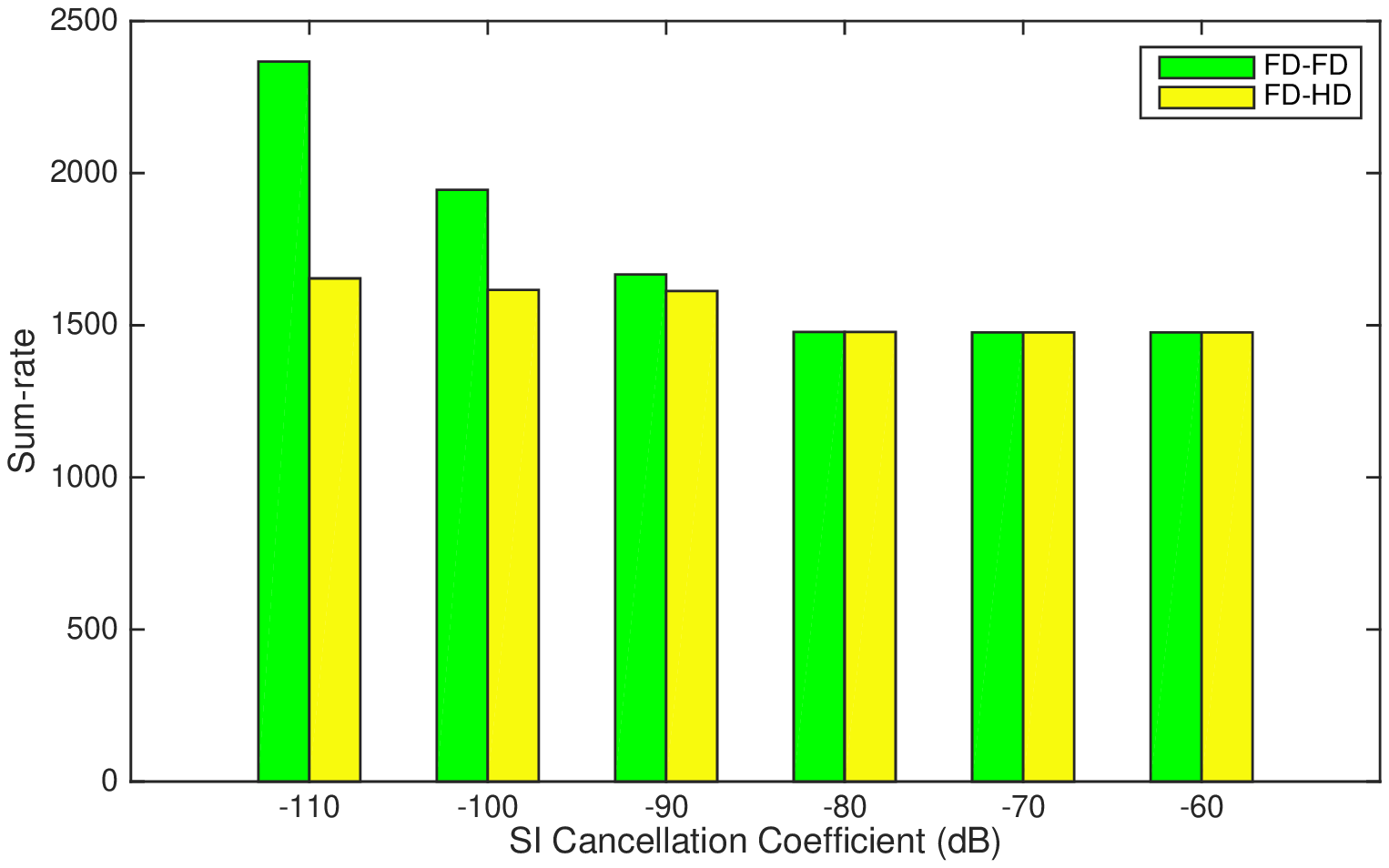}
\caption{\small{Effect of the self interference cancellation coefficient on the FD network capacity in the indoor scenario.}}
 \label{fig:beta_threshold_indoor}
\end{minipage}
\hfill
\begin{minipage}{\linewidth}
\includegraphics[width=\linewidth]{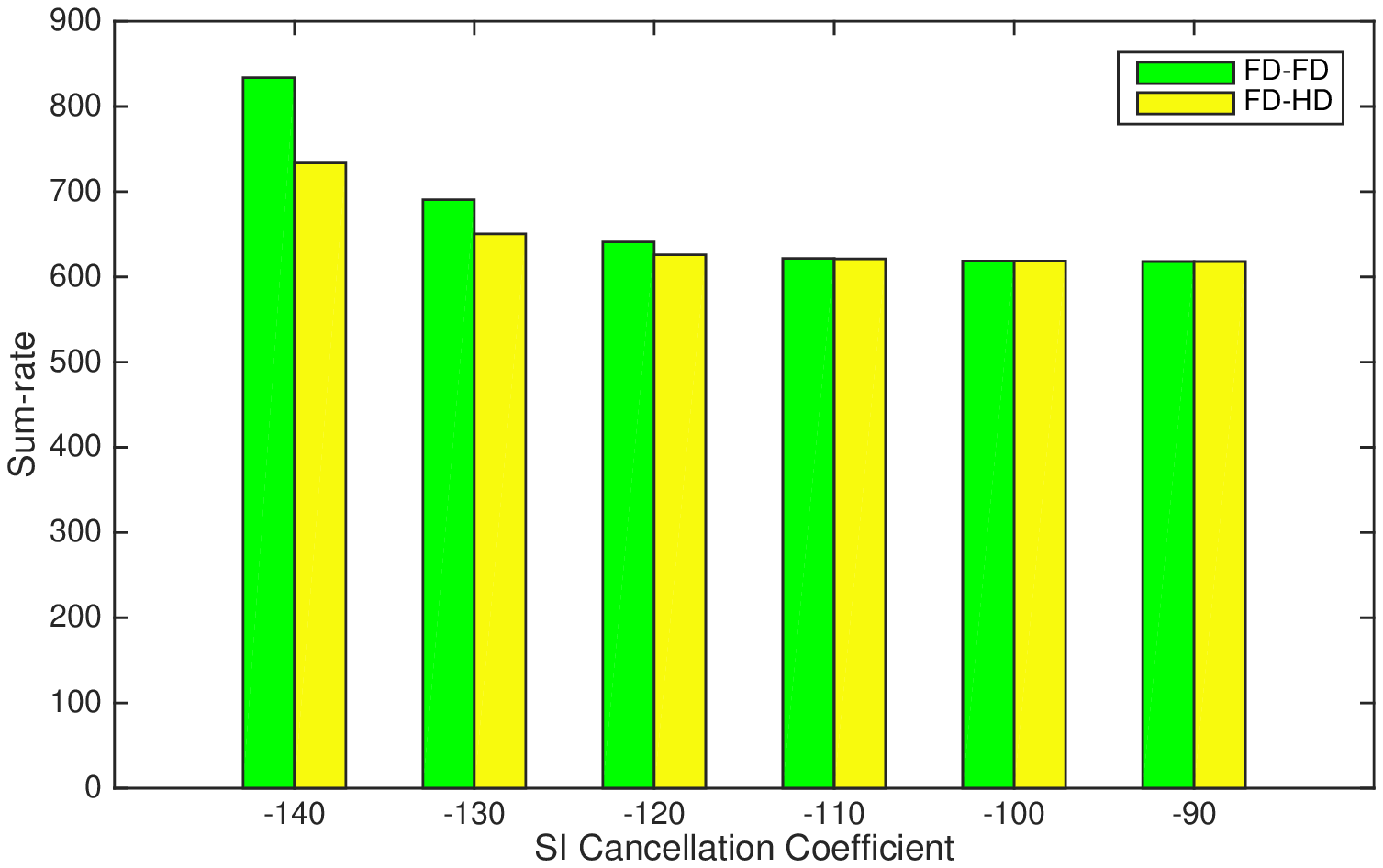}
  \caption{\small{Effect of the self interference cancellation coefficient on the FD network capacity in the outdoor scenario}}
          \label{fig:beta_threshold_outdoor}
\end{minipage}%
\vspace{-8mm}
\end{figure}

 Fig. \ref{fig:beta_threshold_indoor} compares the sum-rates of an FD-FD network and an FD-HD network for different values of $\beta$ in the indoor scenario. It can be seen that when $\beta$ is larger than a specified threshold, which is near $-90$ dB, there is no difference between the sum-rate of the all HD user case and the sum-rate of the all FD user case. The reason is that when $\beta$ is large relative to the inter-node interference, FD users prefer to work in HD mode in order to increase their rate, hence the sum-rates of FD-FD and FD-HD become equal.  

In Fig. \ref{fig:beta_threshold_outdoor}, the same experiment is repeated for the outdoor scenario. Here the threshold $\beta$ is approximately $-120$ dB which is much smaller than in the indoor case. Since the inter-node interference in the outdoor environment is smaller, the SI cancellation coefficient should be very small to make the FD mode worthwhile for the FD users. 

\begin{figure}[t!]
\begin{minipage}{\linewidth}
\includegraphics[width=\linewidth]{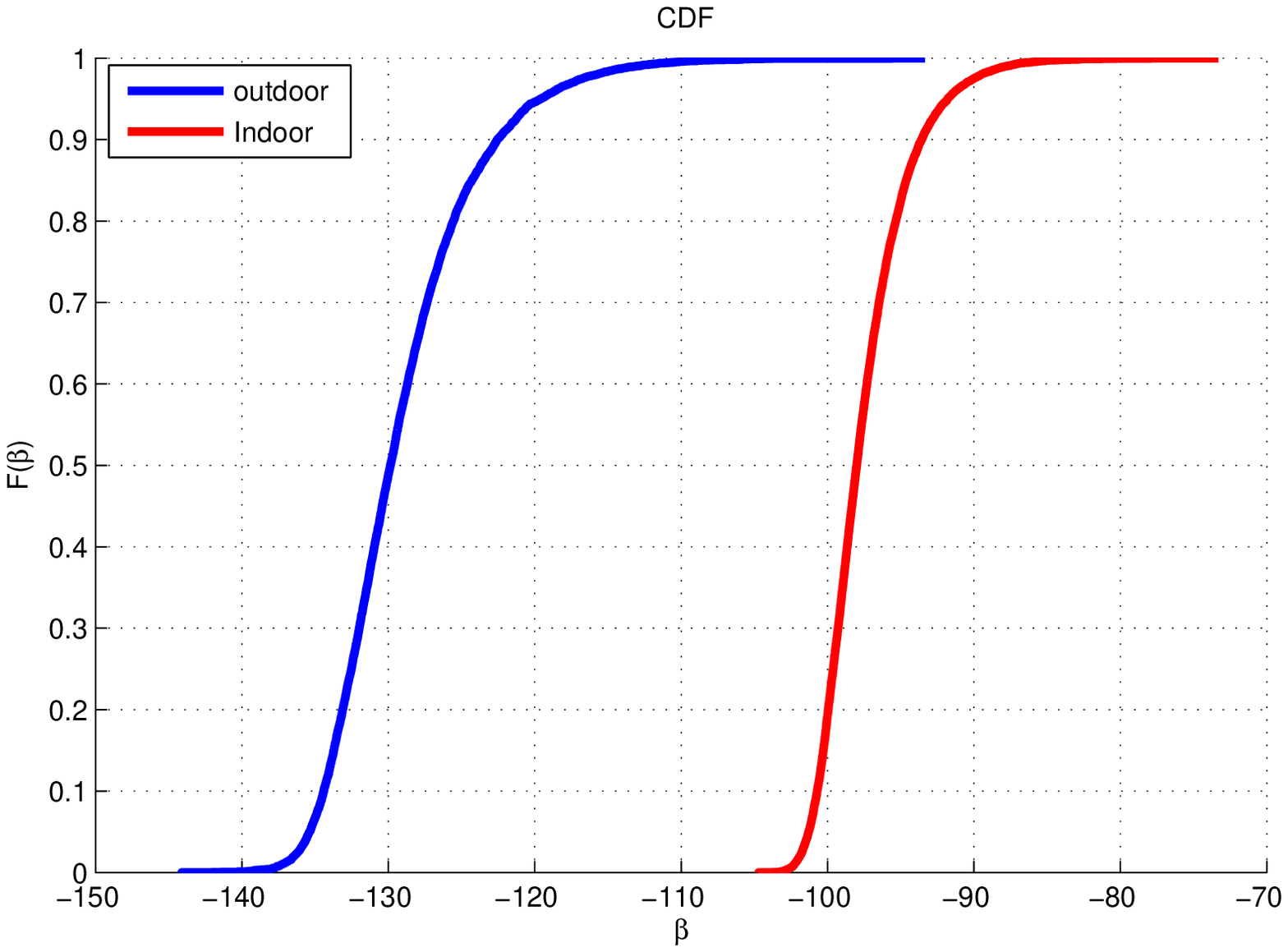}
  \caption{\small{Cumulative Distribution Function of the threshold $\beta$ for the indoor and outdoor environment} }
  \label{fig:CDF_compare}
\end{minipage}
\hfill
\begin{minipage}{\linewidth}
\includegraphics[width=\linewidth]{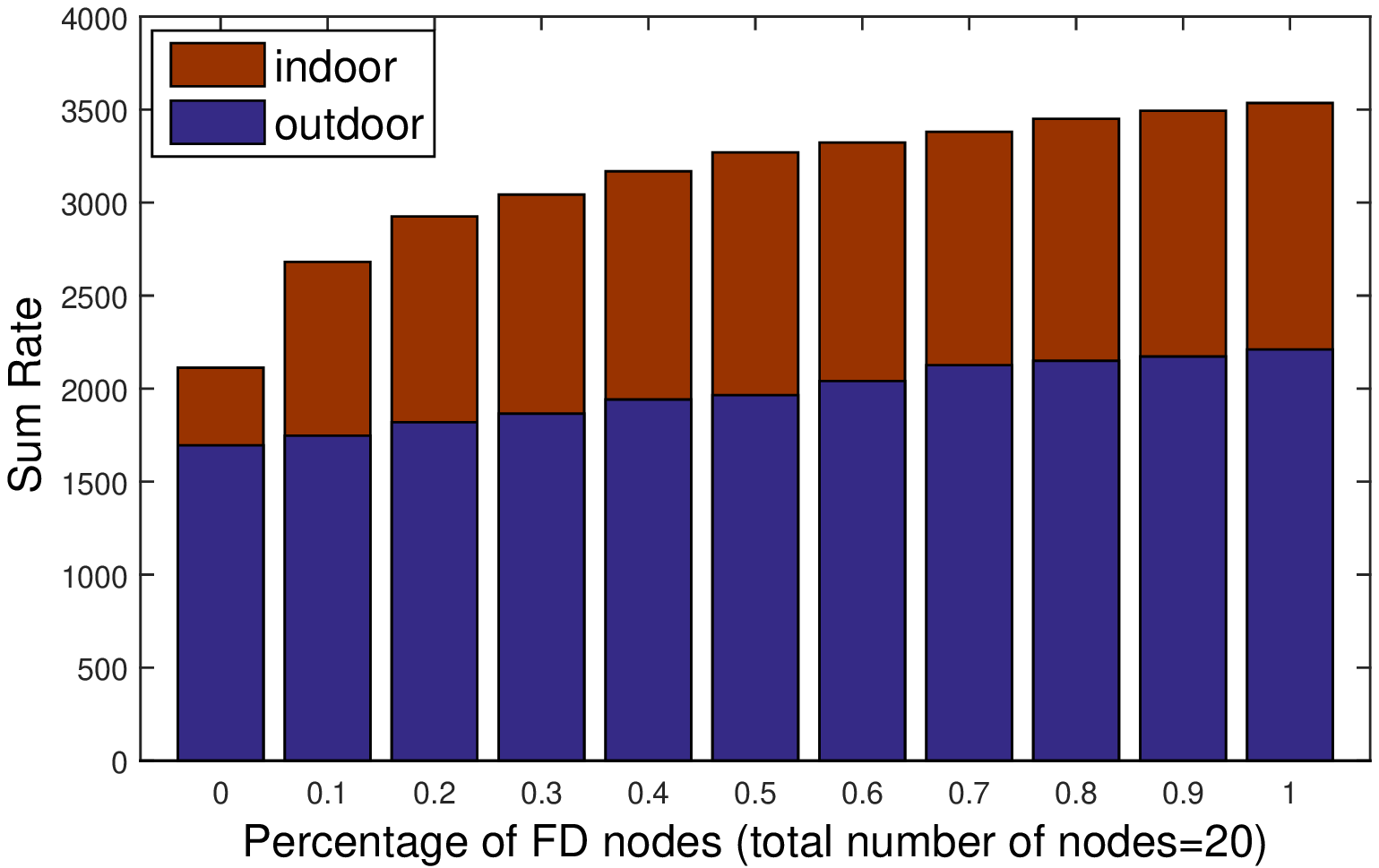}
  \caption{\small{Effect of the fraction of FD nodes on the network throughput in both indoor and outdoor scenarios}}
  \label{fig:FDpercentage}
\end{minipage}%
\vspace{-8mm}
\end{figure}

Fig. \ref{fig:CDF_compare} shows the CDF of the threshold $\beta$ for both indoor and outdoor environments based on the analysis in Proposition 2. As evident, in the indoor and outdoor scenarios, the CDF curve almost reaches one for  a threshold $\beta$ close to  $-90$ dB and $-120$ dB, respectively. These results match those in Fig. \ref{fig:beta_threshold_indoor} and Fig. \ref{fig:beta_threshold_outdoor} obtained through simulations. Therefore, the presented analysis is able to accurately predict the required self-interference cancellation performance.

Fig. \ref{fig:FDpercentage} shows the performance of a full-duplex OFDMA network with a mix of FD and HD users. A total of 20 users are considered, assuming perfect SI cancellation  for FD devices. It can be seen that increasing the percentage of FD users in an outdoor environment does not increase the total sum-rate significantly, but in the indoor case by equipping only $10\%$ of the nodes with FD technology the network throughput greatly increases. The reason behind this is the large inter-node interference in the indoor environment that could be avoided by using FD users instead of HD ones.

\begin{figure}[t!]
\begin{minipage}{\linewidth}
\includegraphics[width=\linewidth]{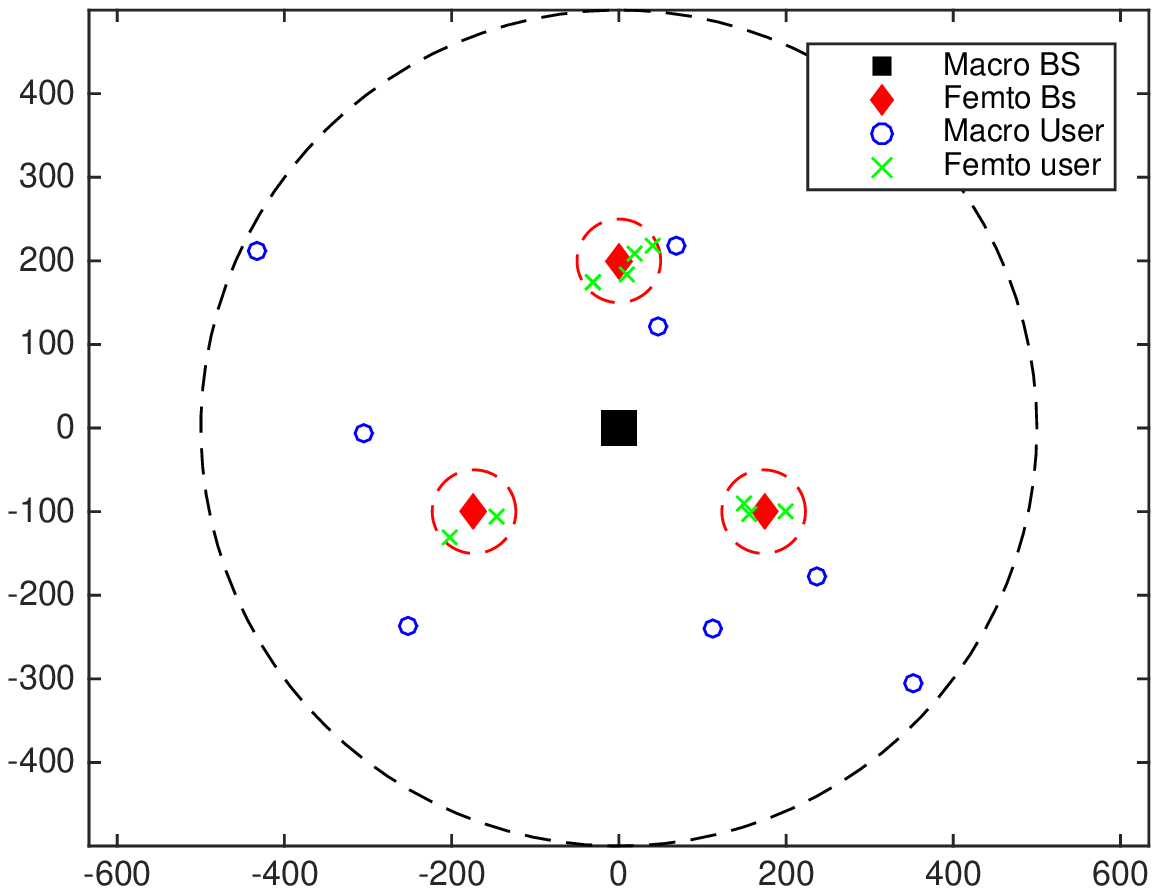}
  \caption{\small{Network topology}}
        \label{fig:topology}
\end{minipage}
\hfill
\begin{minipage}{\linewidth}
\includegraphics[width=\linewidth]{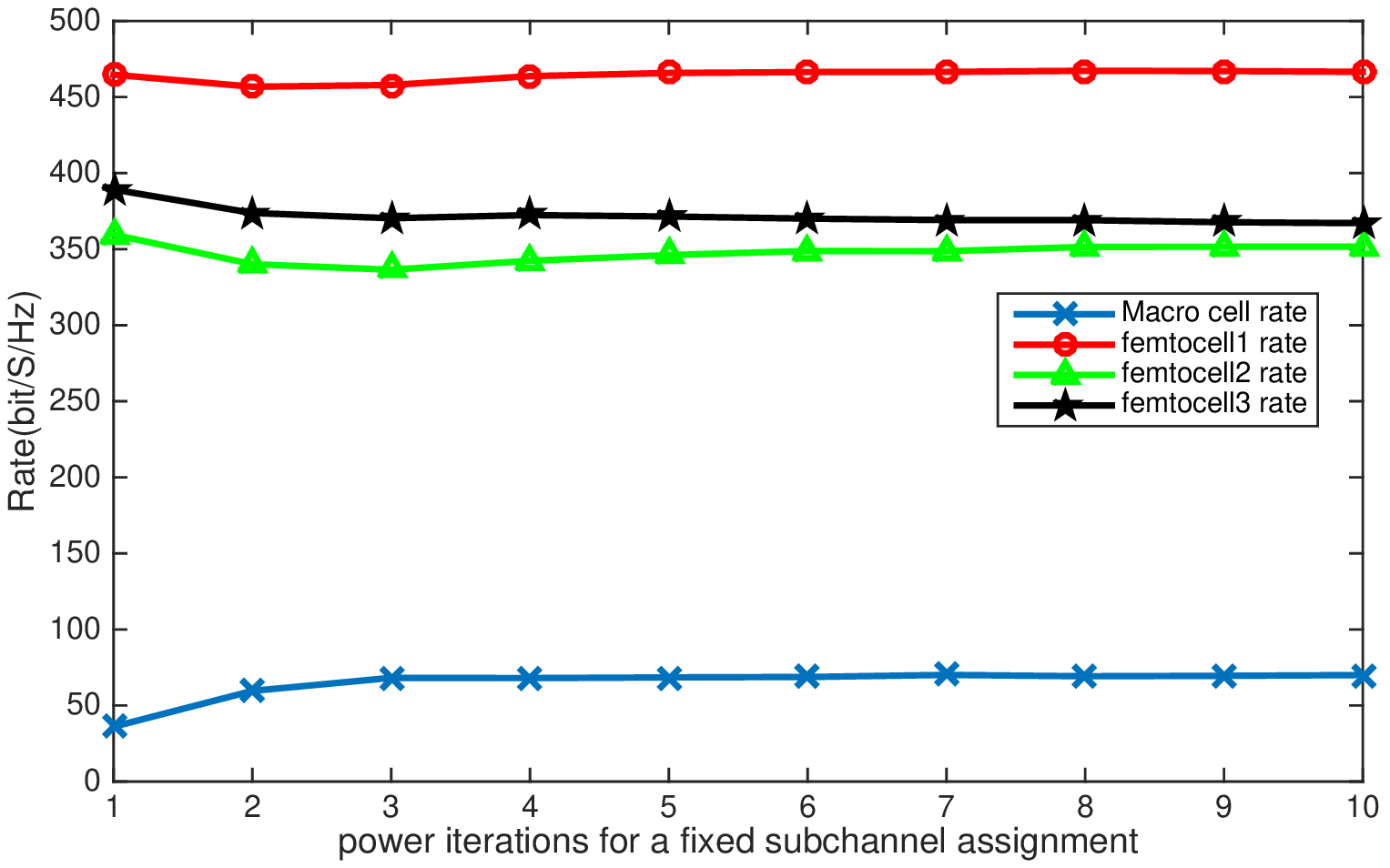}
  \caption{\small{Convergence of the power allocation algorithm}}
        \label{fig:power_iteration_convergence}
\end{minipage}%
\vspace{-8mm}
\end{figure}

For the heterogeneous OFDMA network, we consider a macro BS with 8 users and 3 femto cell BSs with 2, 3, and 4 users, respectively. The location of base stations and users are depicted in Fig. \ref{fig:topology}. We assumed that macro users and femto users are randomly spread within a cell radius of 500 m and 50 m around their related BSs, respectively.

Fig. \ref{fig:power_iteration_convergence} shows the convergence of the power allocation algorithm of Section \ref{power_allocation} for a fixed sub-channel assignment in the said heterogeneous network. Here, the minimum downlink and uplink rates for macro cell users are set to $R_{mind}=35$ bit/s/Hz and $R_{minu}=35$ bit/s/Hz, respectively. The total rate of macro users is then to be larger than $70$ bit/s/Hz. As we see in this figure, this constraint is satisfied after a few iterations. Also it was expected that this inequality constraint should be satisfied with equality, since if the total rate of macro users becomes bigger than the constraint it increases extra interference for femto cells and reduce the objective function.

\begin{figure}[t!]
\begin{minipage}{\linewidth}
  \includegraphics[width=\linewidth]{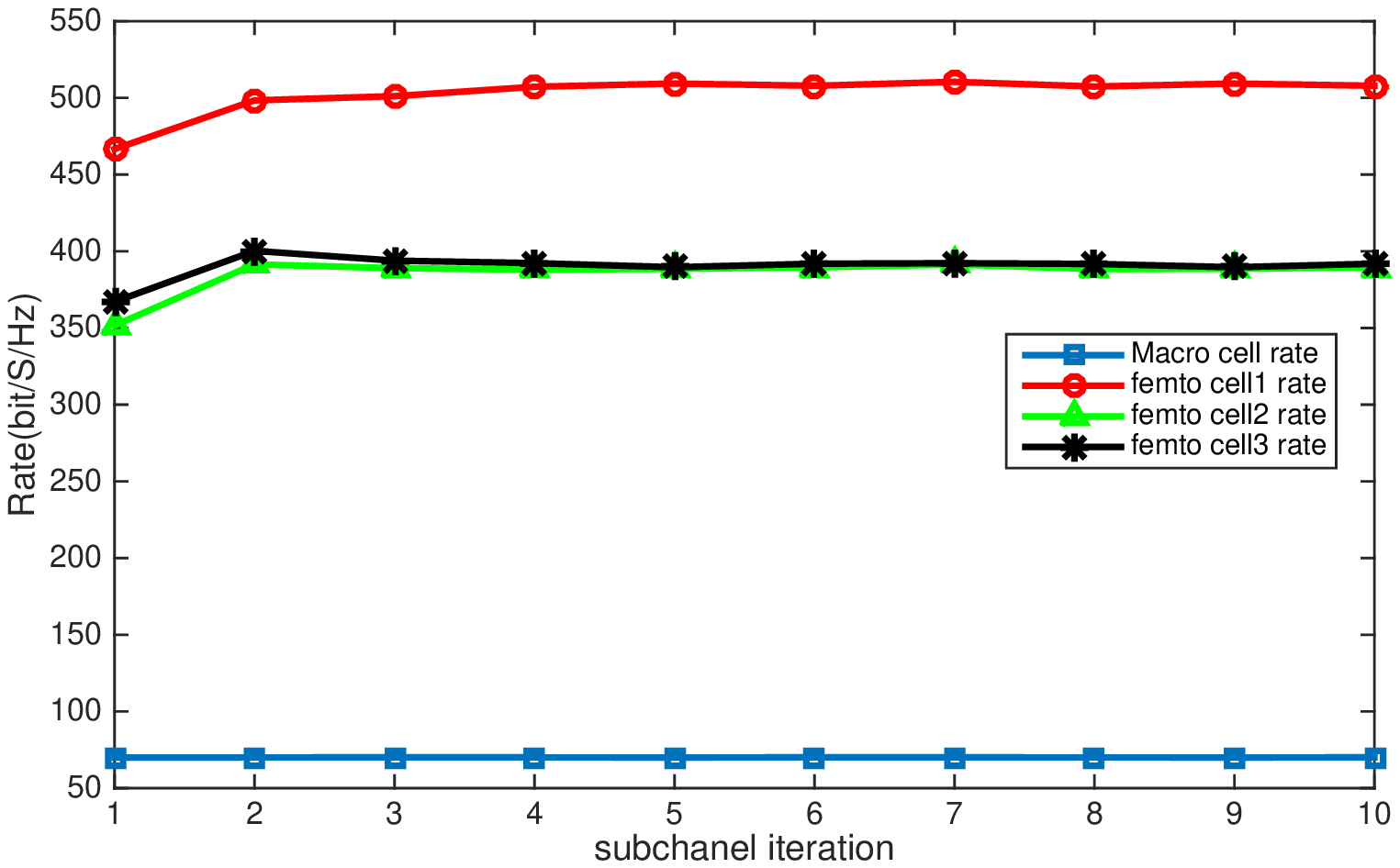}
  \caption{\small{Convergence of the iterative subchannel allocation algorithm}}
        \label{fig:subchannel_iteration_convergence}
\end{minipage}
\hfill
\begin{minipage}{\linewidth}
  \includegraphics[width=\linewidth]{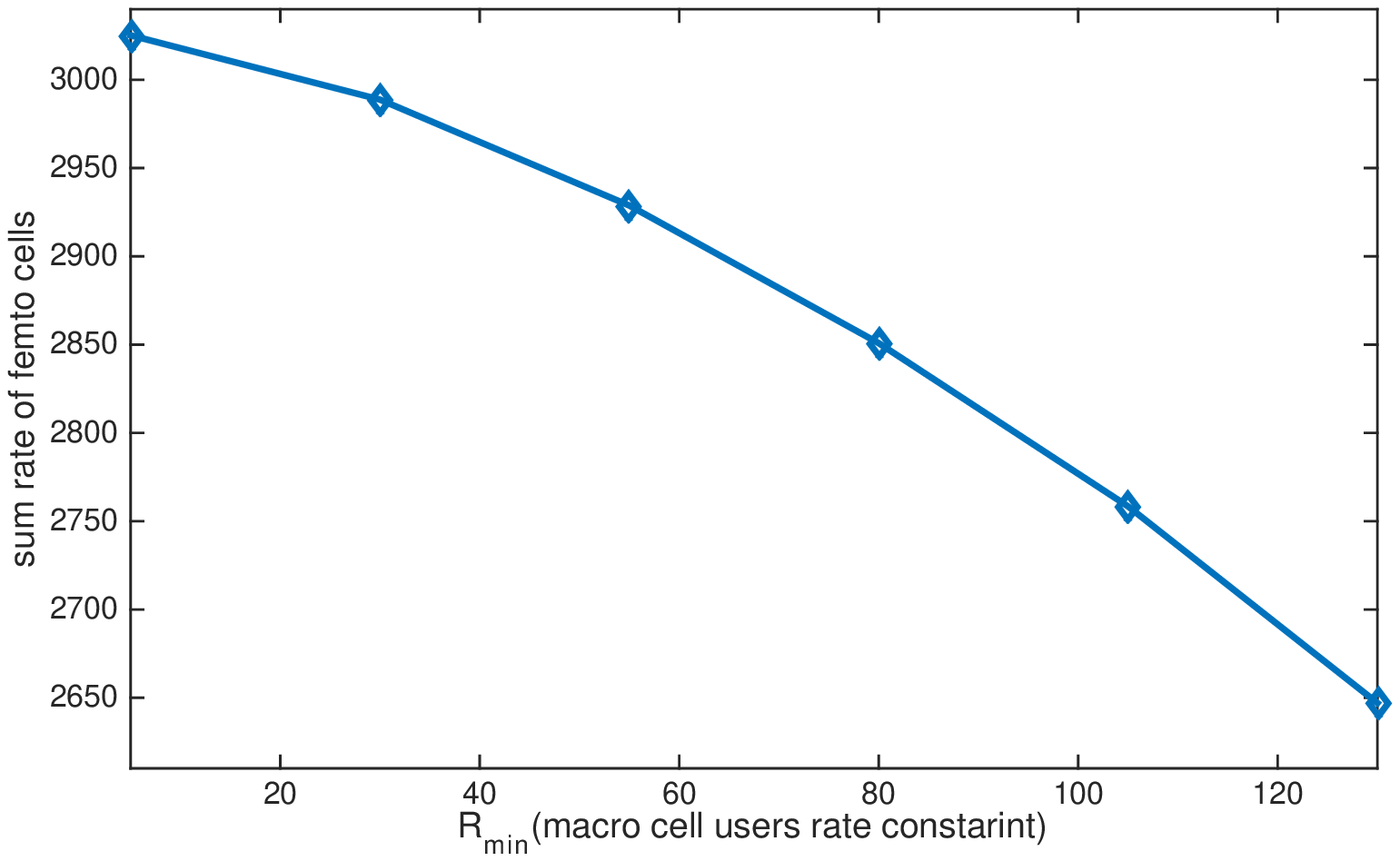}
  \caption{\small{Tradeoff between femto cell throughput and macrocell $R_{min}$}}
        \label{fig:tradof}
\end{minipage}%
\vspace{-8mm}
\end{figure}

Fig. \ref{fig:subchannel_iteration_convergence} shows the convergence of the iterative sub-channel allocation algorithm. It can be seen that the femto cell rates converge after a few iterations, while the minimum rate for macro cell users is satisfied.

Fig. \ref{fig:tradof} depicts the sum-rate of femto cells as a function of the minimum required rate for the macro cell. As evident, the larger the macro cell rate, the smaller the femto cells sum-rate. The reason is that by increasing the rate of the macro cell, the interference caused by macro cell to the femto cell would also increase, thereby reducing the femto cells sum-rate.

\begin{figure}
\centering
  \includegraphics[width=\linewidth]{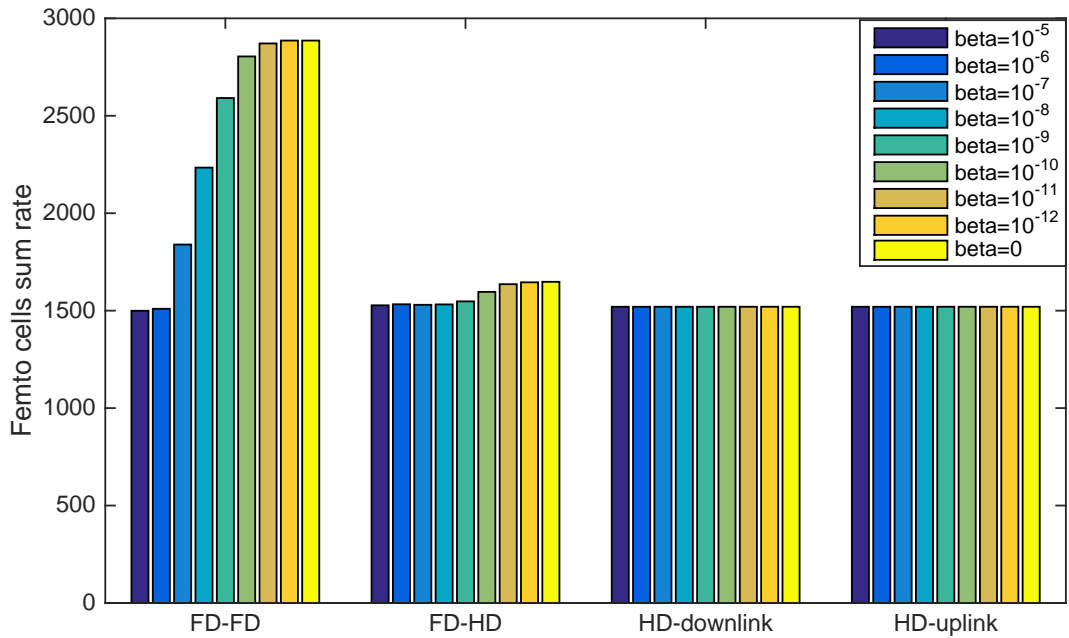}
  \caption{\small{femto cells sum rate for four different schemes and different $\beta$s}}
        \label{fig:compareness}
\end{figure}

Fig. \ref{fig:compareness} shows the femto cells sum-rate for different schemes (FD-FD, FD-HD, HD-downlink, HD-uplink) and different self interference cancellation coefficient values. The $\beta$ threshold effect is also obvious in this graph and its value is different from the previous sections because of the adopted path loss model and different distances in the heterogeneous setting. In this case, one cannot double the capacity by FD transmission because of the interference generated by other femto cells. Still, significant gains maybe achieved if users adopt the FD technology.
\pagebreak
  
\section{conclusion}
\label{conclusion}
In order to fully exploit the advantages of FD technology in wireless networks, it is important to design  appropriate resource allocation algorithms that consider the FD capability of the nodes and the BSs. In this paper, first we considered a single cell OFDMA network that contains an FD BS and a mixture of HD and FD users, and also assumed that FD nodes are not  necessarily perfect FD devices and may suffer from residual self-interference. For this model, we proposed a sub-channel allocation algorithm and power allocation method and showed that when all users and the BS have perfect FD transceivers, we can double the capacity. Otherwise, because of inter-node interference and self-interference the spectral efficiency gain is smaller, but we showed that even by using an imperfect FD BS in a network, the network throughput could  increase significantly. Then, we used the extended version of the proposed algorithms to solve an optimization problem for an FD OFDMA heterogeneous  network in which inter-cell interference should be taken into account. We also investigated FD operation in both outdoor and indoor scenarios and studied the effect of the self interference cancellation coefficient and of the percentage of FD users. Finally, we analyzed the effect of the SI cancellation level on the network performance and numerically computed the CDF of the SI cancellation coefficient threshold which had been observed in the simulation results.

\bibliographystyle{IEEEtran}
{\small
\bibliography{bibfiles}}
\end{document}